\documentclass[11pt,draftclsnofoot,peerreviewca,onecolumn,letterpaper]{IEEEtran}

\advance \topmargin by -\headheight
\advance \topmargin by -\headsep
\evensidemargin \oddsidemargin
\usepackage{amsmath}
\usepackage{amssymb,theorem}
\usepackage[pdftex]{graphics}
\usepackage{epsfig}
\usepackage[pdftex]{color}
\usepackage{psfrag,bbm}
\usepackage{relsize}

\usepackage{latexsym}
\usepackage{epsfig}

\usepackage{graphicx}
\usepackage[font=small]{caption}
\usepackage[font=footnotesize]{subcaption}
\usepackage{wrapfig}
\usepackage{cite}

\long\def\comment#1{}

\newfont{\bbb}{msbm10 scaled 800}

\newfont{\bb}{msbm10 scaled 1100}
\newcommand{\CC}{\mbox{\bb C}}

\newcommand{\EE}{\mbox{\bb E}}

\newcommand{\sv}{{\bf s}}

\newcommand{\uv}{{\bf u}}

\newcommand{\xv}{{\bf x}}
\newcommand{\yv}{{\bf y}}
\newcommand{\zv}{{\bf z}}

\newcommand{\Hm}{{\bf H}}

\newcommand{\Km}{{\bf K}}

\newtheorem{thm}{Theorem}
\newtheorem{lemma}{Lemma}
\newtheorem{defn}{Definition}
\newtheorem{remark}{\indent \bf Remark}

\newcommand{\msf}[1]{\mathsf{#1}}

\begin{document}

\title{Cooperation Alignment for Distributed Interference Management}

\author{Vasilis~Ntranos$^\dagger$\thanks{This work is the outcome of a collaboration that started while V. Ntranos was a research intern at Bell Labs, Alcatel-Lucent. Emails: ntranos@usc.edu,
mohammadali.maddah-ali@alcatel-lucent.com, caire@tu-berlin.de. A shorter version of this paper containing preliminary results
was presented at IEEE Int. Symp. on Inform. Theory, Hong Kong, June  14-19, 2015}, 
        Mohammad~Ali~Maddah-Ali$^\ast$,  and
        Giuseppe~Caire$^\ddagger$ \\
        $^\dagger$University of Southern California, Los Angeles, CA,  USA\\ 
				$^\ast$Bell Labs, Alcatel-Lucent, Holmdel, NJ, USA \\
				$^\ddagger$Technical University of Berlin,  Germany}
        
\maketitle

\begin{abstract}
We consider a cooperative Gaussian interference channel in which each receiver must decode its intended message 
locally, with the help of cooperation either at the receivers side or at the transmitter side. 
In the case of receiver cooperation, the receivers can process and share information through limited 
capacity backhaul links. In contrast to various previously considered distributed antenna architectures, 
where processing is performed in a centralized fashion, the model considered in this paper aims to capture the essence of decentralized 
processing, allowing for a more general class of ``interactive'' interference 
management strategies.  For the three-user case, we characterize the fundamental tradeoff between the achievable 
communication rates and the corresponding backhaul cooperation rate, in terms of degrees of freedom (DoF). 
Surprisingly, we show that the optimum communication-cooperation tradeoff per user remains the same when we move from the two-user to three-user case.
In the absence of cooperation, this is due to interference alignment. 
When backhaul cooperation is available,  we develop the new idea of \emph{cooperation alignment}, which guarantees that the 
average (per user) backhaul load remains the same as we increase the number of users.  
In the case of transmitter cooperation, the transmitters can form their jointly precoded signals through an interactive protocol over the backhaul. 
Specifically, we show that the optimal (per user) communication/cooperation tradeoff in the three-user case is the same 
as for receiver cooperation.
\end{abstract}

\newpage


\section{Introduction} \label{intro}
Consider a $K$-user Gaussian interference channel with cooperation either at the transmitter or at the receiver side. This paper focuses on the 
fundamental limits of distributed cooperation, achieved through a wired backhaul network consisting of noiseless links, 
through which every pair of receivers (resp., transmitters)
can interact.
A natural question arising from this cooperative communication setup is ``How much backhaul capacity is required in order to 
achieve a given communication rate?'' or, equivalently ``What is the best communication rate that one can achieve for a given constraint 
on the total backhaul capacity?''  

We first focus on {\em receiver cooperation}.  
In this case, for the two-user case the \emph{communication vs cooperation tradeoff} has been characterized within a constant gap  
in \cite{wt-Rx11}. From a degrees of freedom (DoF) perspective, if the average (per user) rate scales as $R = \msf{DoF}\cdot\log(P) + o(\log(P))$ 
and the average (per user) backhaul cooperation load scales as $L = \alpha\cdot\log(P)+ o(\log(P))$, the results of 
\cite{wt-Rx11} can be used to show that the optimal communication vs cooperation tradeoff for the two-user interference channel 
is given by $\msf{DoF}^{*}(\alpha)=\min\{1,\frac{1+\alpha}{2}\}$.
This is a very intuitive result in terms of the achievable {DoF}; when $\alpha=0$, one can  achieve $\msf{DoF}(0)=1/2$ by orthogonal user 
scheduling and when $\alpha=1$, $\msf{DoF}(1)=1$ can easily be achieved by exchanging the user's (appropriately quantized) received 
observations over the backhaul, such that each receiver has two observations and can eliminate (e.g., by simple linear processing) the unintended signal interference. However, we can immediately see that following the same approach for the $K$-user  case is not optimal in general. 
To begin with, it is well-known that transmission schemes based on interference alignment \cite{mmk08,cj08,mgmk09} are still able to 
achieve $\msf{DoF}(0) = 1/2$ no matter how many users are interfering in the network. 
The fundamental question that we aim to answer in this work is whether the same holds for all  values of $\alpha\geq0$; or, to put it in other words, 
whether the  entire communication vs cooperation tradeoff remains unaffected by the presence of more than two interfering links.

Surprisingly,  our results  show  that this is indeed the case for the three-user interference channel. This result is shown in this paper by developing  
the new idea of \emph{cooperation alignment} that has the same effect on the backhaul load as interference alignment has on the ``wireless'' 
degrees of freedom; from each receiver's perspective, it appears as if a \emph{single} user jointly processes the observations of the entire 
network and only shares the necessary information over the backhaul. 

In order to explain the idea of cooperation alignment let us focus on the case of $\alpha=1$ with $\msf{DoF}(1)=1$ in the noiseless case, which captures the essence of degrees of freedom.  In the three-user interference channel, receivers one, two, and three, observe the interfering terms $h_{12}x_2+h_{13}x_3$, $h_{21}x_1+h_{23}x_3$, and $h_{31}x_1+h_{32}x_2$, respectively. 
If each receiver had access to its own interference,\footnote{It is sufficient that the interference term is known within a distortion with bounded 
mean-square error as the signal power increases.} 
then it would be able cancel it out, and decode its own message with DoF of one.  Given that the transmit power is $P$, this knowledge itself would require about $3\log(P)+o(\log(P))$ bits of information and each receiver would need to receive at least $\log(P)+o(\log(P))$ bits from the backhaul. 
However, the challenge is that these exact interfering signals are not available at any of the receivers. For example, none of the receivers has access to  the particular combination of $h_{12}x_2+h_{13}x_3$ that receiver one needs in order to decode its  message. Therefore, it would seem impossible to achieve one DoF per user in the wireless channel, with a (per user) backhaul load of  $\log(P)+o(\log(P))$. In this paper, however, we show  that with \emph{cooperation alignment}, the receivers are able to create these combinations in a distributed manner, through an iterative process in which they sequentially decode small parts of their original messages and share ``carefully chosen'' interfering combinations over the backhaul.

For the case of transmitter cooperation, we consider the same backhaul connectivity model as in the case of receiver cooperation but with the role 
of transmitters and receivers being exchanged. Namely, the transmitters are allowed to cooperate by exchanging backhaul messages with the purpose
of jointly encoding their transmitted signals to mitigate interference, while the receivers will attempt to decode their intended messages solely 
based on their received observations.\footnote{While in this paper we consider the problem from a purely information-theoretic viewpoint, 
it is clear that, in practice, receiver and transmitter cooperation are relevant to the uplink and downlink of a cellular/wireless network, respectively, 
where the cooperation through the backhaul network is implemented in both cases at the base-station side.}
For this case, we show that the optimal tradeoff $\msf{DoF}^{*}(\alpha)$ in the case of three users is the same as for the case of transmitter cooperation.

The main contributions of this papers are as follows.
\begin{itemize}
\item We propose an information theoretic  channel model that reveals the  fundamental challenges of decentralized (over the cloud) backhaul cooperation in wireless  networks.
\item We characterize the optimum communication vs cooperation tradeoff for three-user interference channels.
\item We exhibit a new form of alignment  -- that we term \emph{cooperation alignment} -- 
that is able to achieve the optimal tradeoff.
\end{itemize}

{\it Related Work:}
Several results, that have developed and used techniques that are closely related~to  our achievable schemes, can be found in \cite{nnw13,nw12,sa14,ng11,mgmk09}.
In particular, \cite{nnw13,nw12} proposed a lattice coding scheme for compute-and-forward \cite{ng11} based on  techniques developed for interference alignment \cite{mgmk09,cj08,ergodic} to show that $K$ relays can reliably decode a (jointly) invertible function of the $K$ interfering messages sent by the transmitters and achieve a computation rate with $K$ degrees of freedom. More recently, \cite{sa14} focused on a $K\times K\times K$  (two-hop) wireless relay network and used similar techniques to design a novel aligned network diagonalization scheme that is able to distributedly invert the corresponding decoded functions of the $K$ transmitted messages, over the second wireless channel from the $K$ relays  to the $K$ receivers.
It is important also to note that  distributed interference management techniques have been considered in  \cite{ssps08-d,nmc14,Cell_IA_omni} in the context of cellular networks, under different backhaul cooperation models.

{This paper is organized as follows.} For ease of exposition and in order to avoid overly repetitive definitions, 
we first focus on receiver cooperation and then extend the problem definition to the case of transmitter cooperation.  In  
Section~\ref{sec:prob} we provide the basic definitions and formally describe the proposed channel model. 
Then, in Section~\ref{sec:main} we outline our main results for the communication vs cooperation tradeoff and in Section~\ref{achiev} 
and Appendix~\ref{proof:conv}  we give the corresponding proofs. 
In Section \ref{transmitter-cooperation-section} we focus on transmitter cooperation, where we develop the corresponding problem definition
and provide analogous results (i.e., the characterization of the optimal communication vs cooperation) for this case. 
Finally, we conclude this paper with Section~\ref{conclusions}.

\section{Cooperation at the Receiver Side: Problem Statement}\label{sec:prob}

\subsection{Distributed Cooperation Channel Model}\label{sec:modelRX}
The channel model considered in this paper is illustrated in Fig.~\ref{channel}.
Each transmitter $i\in \{1,...,K\}$ has a message $W_{i}$ (intended for receiver $i$) which is encoded into a block \mbox{length-$n$} codeword $[x_{i}(t)]_{t=1}^{n}$  satisfying the average power constraint $\frac{1}{n}\sum_{t=1}^{n}|x_{i}(t)|^{2}\leq P.$
The received signal at the $i$th receiver at time $t=1,...,n$ is given by
\begin{equation}
y_{i}(t) = \sum_{j=1}^{K} h_{ij} x_{j}(t) + z_{i}(t),\vspace{-0.05in}\end{equation}
where $h_{ij}\in \CC$ is the (complex) channel gain between the $j$th transmitter and the $i$th receiver, and $z_{i}(t)$ is the additive circularly-symmetric complex Gaussian noise observed at receiver $i$ with zero mean and unit variance.

\begin{figure}[ht]

                \centering
                \includegraphics[width=1\columnwidth]{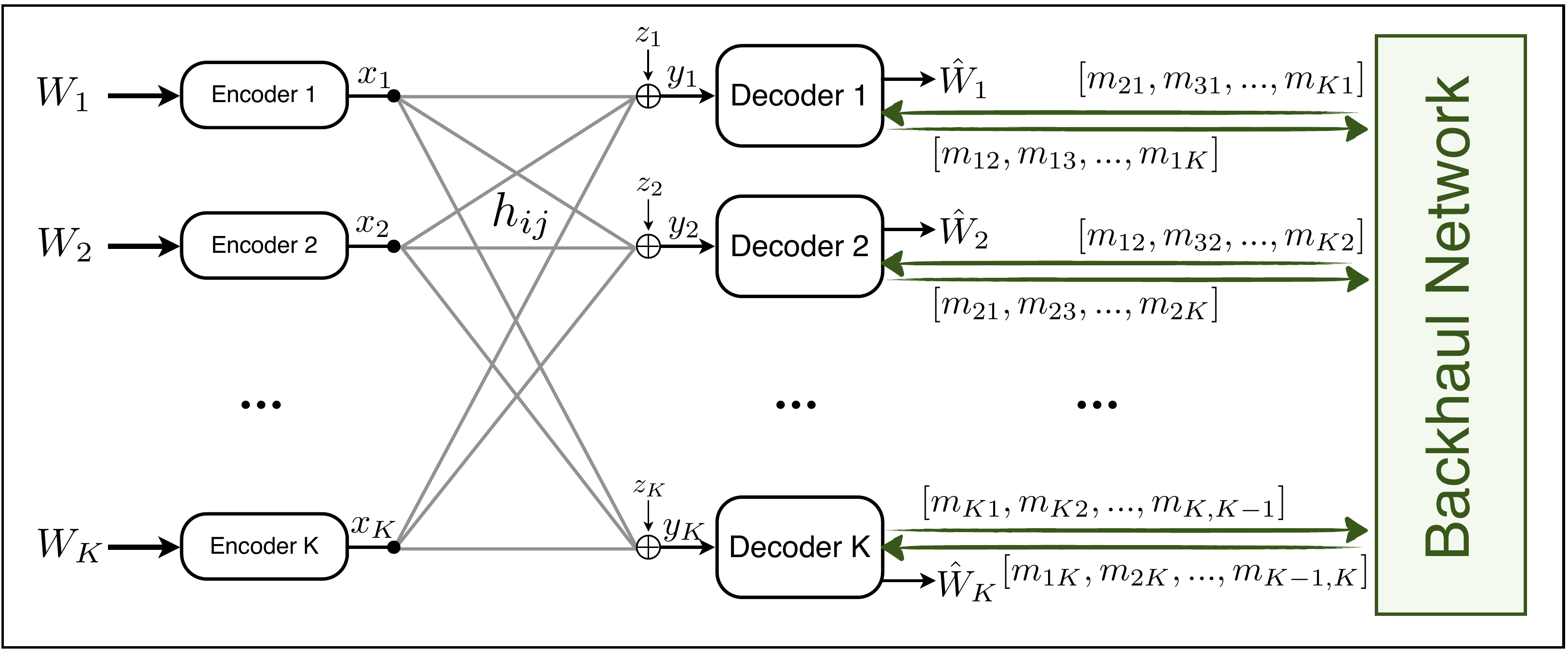}
                \caption{Channel model}
                \label{channel}

\end{figure}

The decoders are able to collaborate over the backhaul  in order to produce their estimates, $\hat W_{i}$, $i\in \{1,...,K\}$. We assume that the backhaul network consists of directed noiseless links $[i,\hat i]$, between every pair of decoders \mbox{$i\neq\hat i$~$\in \{1,...,K\}$}, and the rate from decoder $i$ to  decoder $\hat i$ is denoted by $R_{\rm b}^{[i,\hat i]}$. 
The backhaul message from decoder $i$ to decoder $\hat i$, that passes through the link $[i,\hat i]$ at time $t$, is denoted by $m_{i\rightarrow \hat i}(t)$ and is given as a function of all the previously received signals $[y_{i}(\tau)]_{\tau=1}^{t-1}$ at receiver $i$ and all the previously received messages  $[m_{\ell\rightarrow  i}(\tau)]_{\tau=1}^{t-1}$ from  all decoders $\ell\in \{1,...,K\}, \ell\neq i$.

The rate of each backhaul link $R_{\rm b}^{[i,\hat i]}$ is therefore determined by the average joint entropy of the messages $[m_{i\rightarrow \hat i}(\tau)]_{\tau=1}^{n}$ that pass through it and  is given by
\begin{equation}  \label{ziopino}
R_{\rm b}^{[i,\hat i]} = \frac{1}{n}H\bigl([m_{i\rightarrow \hat i}(\tau)]_{\tau=1}^{n}\bigr).
\end{equation} 
%


An important quantity that will be used in the rest of this paper is the average (per user) backhaul cooperation rate given~by 
\begin{equation}  \label{ziapina}
 \overline R_{\rm b} \triangleq \frac{1}{K}\sum_{i=1}^{K}\sum_{\hat i\neq i}R_{\rm b}^{[i,\hat i]}.
\end{equation}

%
%
%
\subsection{Achievable Rates, Capacity, and Degrees of Freedom}

The rate tuple  $(R_{1},R_{2},...,R_{K})$ is called \emph{achievable} under an average backhaul cooperation rate constraint $\overline R_{\rm b}\leq L$, if for any $\epsilon>0$ and sufficiently large $n$, there exist a length-$n$ coding scheme defined by:
\begin{itemize}
\item The message sets ${\cal W}_{i} = \{1,2,...,2^{nR_{i}}\}$, $i=1,...,K$.
\item The encoding functions ${f}_{i}: {\cal W}_{i}\rightarrow \CC^{n}$, $i=1,...,K$.
\item The backhaul relaying functions 
 $g_t^{[i,\hat i]}$ that generate $m_{i\rightarrow\hat i}(t)$ such that  
$$m_{i\rightarrow \hat i}(t) = g_t^{[i,\hat i]} \left([y_{i}(\tau)]_{\tau=1}^{t-1},M_{i}^{t-1}\right) \in {\cal B}^{[i,\hat i]},$$
where  $M_{i}^{t}\triangleq \left\{[m_{\ell\rightarrow  i}(\tau)]_{\tau=1}^{t} : \ell = 1,\ldots,K,\ell \neq i  \right\}$ is the collection of all the backhaul 
messages $m_{\ell\rightarrow i}(\tau)$ received at decoder $i$ up to time $t$, and ${\cal B}^{[i,\hat i]}$ is a finite set that denotes the  
message alphabet used for the backhaul link $[i,\hat i]$. 
\item  The decoding functions 
\[ \eta_i : \mathbb{C}^{n} \times  \prod_{\substack{\ell = 1,\ldots,K \\ \ell\neq i}} 
 \left ({\cal B}^{[\ell, i]}\right )^n \rightarrow {\cal W}_{i},\;
\mbox{ that give }\; \hat W_{i}\triangleq \eta_i \left([y_{i}(\tau)]_{\tau=1}^{n}, M_{i}^{n} \right), \]
\end{itemize}
such that the corresponding probability of error given by $P_{e}^{(n)} \triangleq\mathbb{P}\left(\bigcup_{i = 1}^K \{\hat W_{i}\neq W_{i}\} \right)$ is less~than~$\epsilon$, and the average backhaul cooperation rate satisfies the backhaul load constraint 
$$\overline R_{\rm b}=\frac{1}{K}\sum_{i=1}^{K}\sum_{\hat i\neq i}\frac{1}{n}H \bigl([m_{i\rightarrow \hat i}(\tau)]_{\tau=1}^{n}\bigr)\leq L.$$

\begin{defn}[Capacity Region]
The  \emph{capacity region} ${\mathcal C}_L$ is defined as the closure of the set of all the rate tuples $(R_{1},R_{2},...,R_{K})$ that are achievable with an average backhaul cooperation rate $\overline R_{\rm b}\leq L$.
\end{defn}

\begin{remark}
The region ${\mathcal C}_{0}$ coincides with the capacity region of the $K$-user Gaussian interference channel (no cooperation) and the region ${\mathcal C}_{\infty}$ with the capacity region of the $K$-user Gaussian MIMO multiple access channel with $K$ receive antennas (full cooperation).
\end{remark}

As the transmit power $P$ increases, it is reasonable to let also the backhaul rate constraint $L$ increase as some function $L(P)$. 
Then, for an achievable scheme  we are interested in characterizing the tradeoff between the 
average (per user) backhaul cooperation load given by 
\begin{equation}
\alpha\triangleq \lim_{P\rightarrow\infty } \frac{L(P)}{\log(P)},
\end{equation}
and the average (per user) achievable degrees of freedom (DoF) given by 
\begin{equation}
 \msf{DoF}(\alpha) \triangleq \liminf_{P\rightarrow\infty}  \frac{1}{K\log(P)}\sum_{k=1}^{K} R_{K}.
\end{equation}

The  \emph{average DoF (per user)} of the channel is denoted by $\msf{DoF}^{*}(\alpha)$ and  defined as the supremum of $\msf{DoF}(\alpha)$. 
%

{
\begin{remark}
Notice that when $\alpha = 0$, 
the average degrees of freedom 
$ \msf{DoF}(0) = 1/2$
can be achieved (without any cooperation) by interference alignment. On the other hand, when $\alpha = \infty$ (full cooperation) the average degrees of freedom is
$\msf{DoF}(\infty) = 1$
can be achieved by jointly decoding the $K$ received observations.
\end{remark}
}

\subsection{Example: Centralized processing}
\label{sec:central}

%
%
%
%
%


Under this framework, we can designate a specific receiver, say receiver~1, to take the role of the centralized processor and let all the other receivers quantize (within a constant distortion) and forward their observations to it. Now receiver~1 can jointly process all the observations to decode both its own message and the other the $K-1$ messages and subsequently forward the $K-1$ decoded messages back to their intended receivers (see Fig.~\ref{scheme2}). As we can see this scheme is able to achieve the full  DoF  of $1$ with backhaul cooperation load $\alpha=2\frac{(K-1)}{K}$. If we time-share between this scheme and the asymptotic interference alignment scheme that can achieve $\msf{DoF}(0)=1/2$,  we can obtain  the boundary shown with the dashed line in Fig.~\ref{fig:dofalpha}.
\begin{figure}[ht]

                \centering
                \includegraphics[width=0.8\columnwidth]{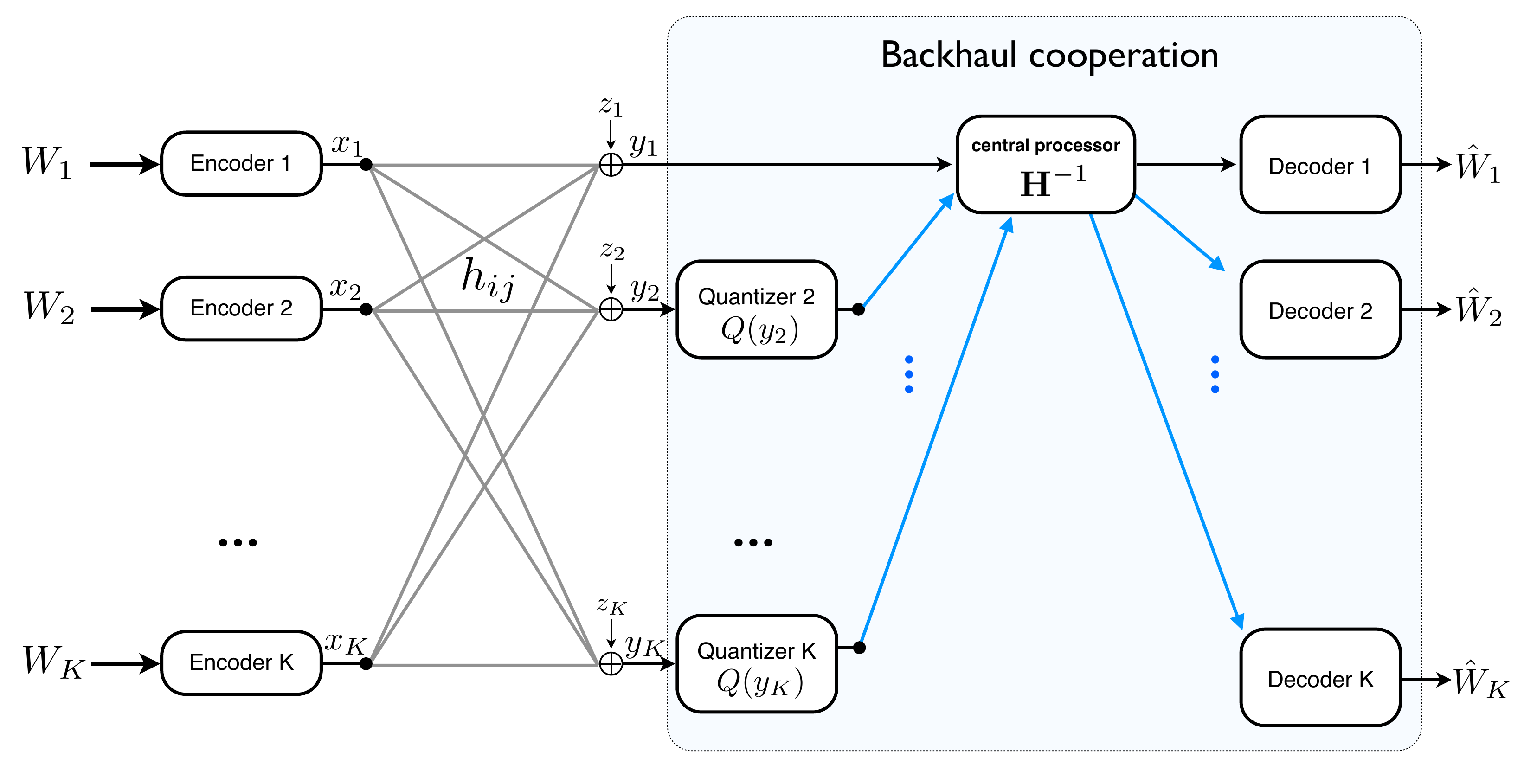}
                \caption{A simple scheme to achieve $1$ DoF per user with $\alpha=2\frac{(K-1)}{K}$}
                \label{scheme2}

\end{figure}

\section{Cooperation at the Receiver Side: Main Results}\label{sec:main}

Our main results on the communication vs cooperation tradeoff for the channel model introduced in the previous section, are described in the following theorems.


\begin{thm}[Upper Bound]
\label{thm:out}
 In the $K$-user interference channel with average backhaul load $\alpha$, we have that  \vspace{-0.05in}
\begin{equation}\msf{DoF}^*(\alpha) \leq \min\left\{1,\frac{1+\alpha}{2}\right\}.\end{equation}
\end{thm}

This outer-bound is derived based on considering every pair of links in the network, and 
developing a bound on communication versus cooperation tradeoff between these two, while the remaining links are 
effectively eliminated from the system (by a genie giving their messages to both receivers). We refer the reader to Appendix~\ref{proof:conv} for the detailed proof.

\begin{remark}Notice that Theorem~\ref{thm:out} shows that for the $K$-user interference channel with average backhaul load $\displaystyle\alpha= \lim_{P\rightarrow\infty } {L(P)}/{\log(P)}=0$, the per user DoF are bounded by $\msf{DoF}^*(0)\leq 1/2$. This matches the well known DoF outer bound for the $K$-user interference channel (without cooperation), and implicitly shows that even when cooperation rates are allowed to scale as $L(P)=o(\log(P))$,  there is no cooperation scheme that can increase the DoF achievable by interference alignment.    
\end{remark}

\begin{thm}[Achievability]  
\label{thm:ach}
In the three-user interference channel with average backhaul load $\alpha$,  we have that \vspace{-0.05in}
 \begin{equation}
 \msf{DoF}^*(\alpha) \geq \min\left\{1,\frac{1+\alpha}{2}\right\}.
 \end{equation}
\end{thm}
\vspace{-0.1in}
\begin{figure}[ht]
                \centering
                \includegraphics[width=0.8\columnwidth]{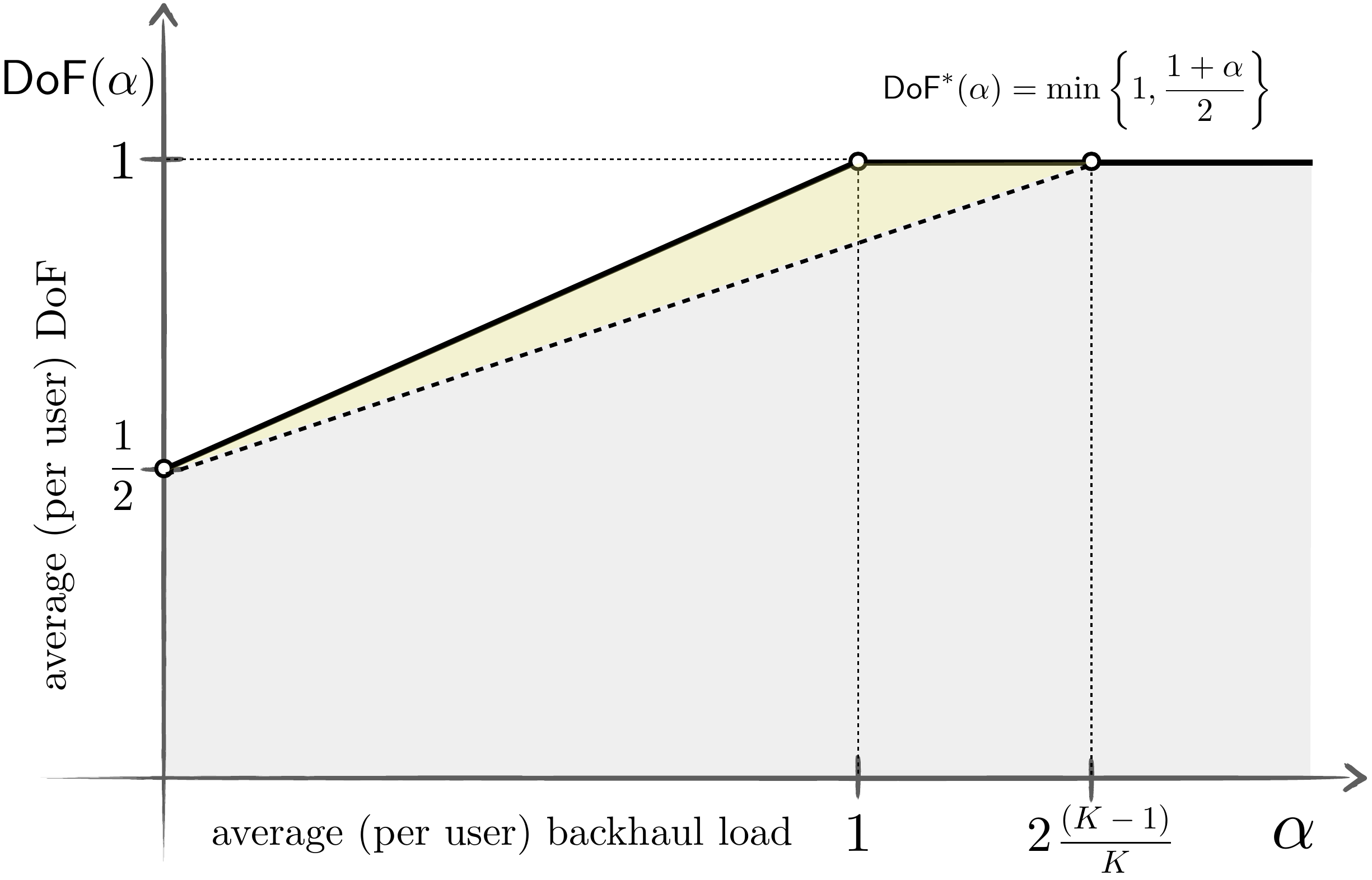}
                \caption{Achievable DoF vs $\alpha$: The dashed line corresponds to the achievable tradeoff for $K$-user channels by centralized processing (Section~\ref{sec:central}) and the yellow  region shows the corresponding gap from the $K$-user outer bound of Theorem~\ref{thm:out}. In Theorem~\ref{thm:ach} we show that \emph{cooperation alignment} is able to close this gap and achieve the optimal communication vs cooperation tradeoff for $K=3$.}
                \label{fig:dofalpha}
\end{figure}

Theorems~\ref{thm:out} and \ref{thm:ach} characterize  $\msf{DoF}^*(\alpha)$ for the three-user interference channel as   $\msf{DoF}^*(\alpha) = \min\left\{1,\frac{1+\alpha}{2}\right\}$.
Recall that, as observed in the introduction, for two-user interference channels, the optimum DoF per user is again $\msf{DoF}^*(\alpha) = \min\left\{1,\frac{1+\alpha}{2}\right\}$, and it is achievable by time-sharing orthogonal access (for $\alpha = 0$) and 
quantize and forward ($\alpha = 1$). 


The question is why for the three-user case, we are able to achieve the same tradeoff.  For $\alpha=0$, the answer is well-understood. We can achieve $1/2$ DoF per user  with interference alignment.  
However, it seems surprising that the DoF per user is the same for other values of $\alpha$. This result basically says that not only in the wireless channel, the extra link does not affect the DoF per user (due to interference alignment), but also in the backhaul, the load of collaboration per user  does not scale with the number of users. The reason is that in the backhaul, we implement another form of alignment,  
 which we call \emph{cooperation alignment}, that is able to  hide the additional collaboration load that is due to the extra link. 
 
\subsection*{Illustrating Example:} 

 To describe the main idea behind  cooperation alignment that we will later use in our achievability proof,  we consider here a specific, yet illustrating, three-user interference channel with $h_{31}=\gamma h_{21}$ and $h_{33}=\gamma h_{23}$, $\gamma\in \CC$, given by
\begin{align}
y_{1} &= \phantom{\gamma}h_{11}x_{1} +  h_{12}x_{2}+  \phantom{\gamma}h_{13}x_{3} +z_{1},\nonumber\\
y_{2} &= \phantom{\gamma} h_{21}x_{1}+ h_{22}x_{2} +  \phantom{\gamma} h_{23}x_{3} +z_{2},\nonumber\\
y_{3} &= \gamma h_{21}x_{1}+ h_{32}x_{2} +  \gamma h_{23}x_{3} +z_{3}.\nonumber
\end{align}
This example is constructed such that the channel coefficients of $x_1$ and $x_3$ are aligned at receivers two and three, which allows us to implement and explain the cooperation alignment scheme in a simple way. 

Let us assume that each transmitter uses a Gaussian codebook, carrying one DoF. We aim to show that each receiver is able to decode its own message, with backhaul load of $\alpha=1$. 
%
In the above example, receiver~3 can first form the backhaul message $m_{3\rightarrow2}$ as the quantized version of $y_{3}$ (with a constant distortion) and send it to
receiver~2. Since $x_1$ and $x_3$ are aligned in $m_{3\rightarrow2}$ and $y_2$,   receiver~2 is able to decode $x_{2}$ at full rate\footnote{Here we assume that a lattice vector quantizer with dither is used, thus  quantization can be modeled with an additive independent quantization noise.} by subtracting $m_{3\rightarrow2}$ from $\gamma y_{2}$.   Notice that at this point receiver~2 can also extract the term $h_{21}x_{1}+h_{23}x_{3}$ from its observation. In order to help receiver~1  decode $x_1$, receiver~2 can now combine $x_2$ and $h_{21}x_{1}+h_{23}x_{3}$ into a single message $m_{2\rightarrow1}$ as the quantized version of 
$h_{12}x_{2}+\frac{h_{13}}{h_{23}}(h_{21}x_{1}+h_{23}x_{3})$.  This combination is formed such that  $x_2$ and $x_3$ in $y_1$ and  $m_{2\rightarrow1}$ are \emph{aligned}.  Sharing $m_{2\rightarrow1}$   over the backhaul will therefore help  receiver 1   decode $x_{1}$ at full rate and subsequently extract the interfering term $h_{12}x_{2}+  h_{13}x_{3}$. In a similar fashion, receiver 1 can recombine the newly available terms $x_{1}$ and $h_{12}x_{2}+  h_{13}x_{3}$ into the  message $m_{1\rightarrow3}$  as the quantized version of  $\gamma h_{21}x_{1} + \frac{h_{32}}{h_{12}}(h_{12}x_{2}+  h_{13}x_{3})$ to help receiver~3 decode $x_{3}$ as well.
This cooperative process, in which receivers iteratively decode desired messages, recombine interfering terms  and share aligned backhaul messages that help another receiver  decode, is what we refer to as \emph{cooperation alignment}. 
\hfill $\lozenge\vspace{0.09in}$

In the above example, we could start the iteration because the combination of the signals  observed at receiver~3 ($\gamma h_{21}x_{1}+\gamma h_{23}x_{3}$) was already aligned with the interference observed at receiver~2. Therefore, by quantizing and forwarding $y_{3}$ to receiver~2, the latter could immediately eliminate its interference and decode $x_{2}$ with $1$ DoF. 
 The major challenge in the three-user IC with generic channel coefficients is that it would have been impossible to 
  find a starting point for the above process since the corresponding channel coefficients are distinct in all receivers with probability~$1$.
 However -- as we will see in the next section-- we are able to  create this form of alignment asymptotically by  
splitting the data streams into many sub-streams, each carrying a small fraction of the total DoF. 
The iterative approach is then started from a vanishing fraction of sub-streams that do not have any interference.

\section{Cooperation Alignment (Proof of Theorem~\ref{thm:ach})}
\label{achiev}

In order to show that $\msf{DoF}(\alpha) = \min\{1,\frac{1+\alpha}{2}\}$, it suffices to show the achievability of the two corner points that give $\msf{DoF}(0) = 1/2$ and $\msf{DoF}(1) = 1$. Then, by time-sharing between the achievable schemes for $\msf{DoF}(0)$ and $\msf{DoF}(1)$ we can obtain the entire boundary $\msf{DoF}(\alpha)=(1-~\hspace{-0.07in}\alpha)\frac{1}{2} + \alpha = \frac{1+\alpha}{2}$, for any $\alpha\in [0,1]$. Further, for any $\alpha>1$, the point $\msf{DoF}(1)$ is trivially achievable and hence one can show that $\msf{DoF}(\alpha) = \min\{1,\frac{1+\alpha}{2}\}$ for all $\alpha\geq 0$.
Since  the point $\msf{DoF}(0) = 1/2$ can be  achieved asymptotically -- without any cooperation -- by interference alignment \cite{cj08,mgmk09}, 
for the rest of this proof we will focus on the achievability of the point $\msf{DoF}(1)=1$, i.e., the achievability of $1$ DoF per user with average backhaul load $\alpha=1$.

\subsection{The achievability of $\msf{DoF}(1) = 1$}\label{proof4uplink}

We first define some short-hand notations. For a natural number $N \in \mathbb{N}$, let $s_{ij}\in \{1,...,N\}$, $i,j\in\{1,2,3\}$, and 
$${\bf s}\triangleq[s_{11},s_{12},s_{13},s_{21},s_{22},s_{23},s_{31},s_{32},s_{33}].$$ 
Then, clearly ${\bf s} \in {\cal S}_{N}\triangleq \{1,..,N\}^{9}$. In addition, for  ${\bf s} \in {\cal S}_{N}$, we define the monomial $\nu_{\bf s}$, as  $\nu_{\bf s}= \prod_{i,j} h_{ij}^{s_{ij}}$. Furthermore, for a positive real number $Q$, we define $\mathbb{Z}_{Q}\triangleq\mathbb{Z}\cap [-Q,Q]$, i.e. 
the integer numbers between $-Q$ and $Q$. 

To form the transmit signals, we use the methodology proposed in~\cite{mgmk09} for real interference alignment, and its extension 
for complex channels in~\cite{mohammad-complex}. In addition, here, we follow the modulation technique used 
in \cite{nw12} to achieve DoF of $K$ for compute and forward (but with different encoding and decoding schemes).
  At transmitter one, the message  $W_{1}$ is split into $N^{9}$ sub-messages, for some $N \in \mathbb{N}$.  Each sub-message is then coded, with  a rate that will be specified later, and modulated over the integer constellation  $\mathbb{Z}_{Q}\triangleq\mathbb{Z}\cap [-Q,Q]$ to form a sub-stream.  Each sub-stream of message $W_1$ is indexed by a unique ${\bf s}\in {\cal S}_{N}$ and denoted by $\{a_{\bf s}(t)\}_{t=1}^n$. The transmitter one at time $t$ sends a  weighted linear combination of sub-streams $a_{\bf s}(t)$, and  scaled by $\Gamma$, as \vspace{-0.05in}
\begin{equation}\textstyle
x_{1}(t)= \Gamma \cdot\sum_{{\bf s}\in {\cal S}_{N}} \nu_{\bf s}  a_{\bf s}(t). \vspace{-0.03in}
\end{equation}
Recall that  $\nu_{\bf s}= \prod_{i,j} h_{ij}^{s_{ij}}$.  The scaling factor $\Gamma$ guarantees that the power constraint is satisfied and will be determined later. We apply the same 
scheme at transmitters two and three to respectively form sub-streams $\{ b_{\bf s}(t)\}_{t=1}^n$ and $\{c_{\bf s} (t)\}_{t=1}^n$, ${\bf s}\in {\cal S}_{N}$, 
and transmit signals $x_{2}(t)= \Gamma\cdot \sum_{{\bf s}\in {\cal S}_{N}} \nu_{\bf s}  b_{\bf s} (t)$ and $x_{3}(t)= \Gamma \cdot\sum_{{\bf s}\in {\cal S}_{N}} \nu_{\bf s}  c_{\bf s}(t)$. 
For simplicity of exposition, in the rest of the proof we drop the time index $t$, unless it is required for clarification.
One can see that, at time $t$, the corresponding received observations for $i\in\{1,2,3\}$ are given by
\begin{equation}\textstyle
y_{i} = \Gamma\cdot \sum_{{\bf s}\in {\cal S}_{N+1}} \nu_{\bf s} \cdot
r_{i,\bf s}  + z_{i}, 
\label{recobs}
\end{equation}
where 
\vspace{-0.26in}
\begin{align}
r_{1,\bf s}\;\triangleq  \;\;&a_{(s_{11}{-1}),s_{12},s_{13},s_{21},s_{22},s_{23},s_{31},s_{32},s_{33}} \nonumber\\
+&b_{s_{11},(s_{12}{-1}),s_{13},s_{21},s_{22},s_{23},s_{31},s_{32},s_{33}}\nonumber\\
+&c_{s_{11},s_{12},(s_{13}{-1}),s_{21},s_{22},s_{23},s_{31},s_{32},s_{33}},
\label{obs1}
\end{align}
\begin{align}
r_{2,\bf s}\;\triangleq  \;\;&a_{s_{11},s_{12},s_{13},(s_{21}{-1}),s_{22},s_{23},s_{31},s_{32},s_{33}}\nonumber\\+
&b_{s_{11},s_{12},s_{13},s_{21},(s_{22}{-1}),s_{23},s_{31},s_{32},s_{33}}\nonumber\\+
&c_{s_{11},s_{12},s_{13},s_{21},s_{22},(s_{23}{-1}),s_{31},s_{32},s_{33}},
\label{obs2}
\end{align}
\begin{align}
r_{3,\bf s}\;\triangleq  \;\;&a_{s_{11},s_{12},s_{13},s_{21},s_{22},s_{23},(s_{31}{-1}),s_{32},s_{33}}\nonumber\\+
&b_{s_{11},s_{12},s_{13},s_{21},s_{22},s_{23},s_{31},(s_{32}{-1}),s_{33}}\nonumber\\+
&c_{s_{11},s_{12},s_{13},s_{21},s_{22},s_{23},s_{31},s_{32},(s_{33}{-1})}. 
\label{obs3}
\end{align}
In addition, for simplicity, in (\ref{recobs}), (\ref{obs1}), (\ref{obs2}), and (\ref{obs3}), for any ${\bf s} \notin {\cal S}_{N}$, we assume that $a_{\bf s}=b_{\bf s}=c_{\bf s}=0$, and we follow this assumption throughout the proof. We further note that $r_{i,\bf s} \in \mathbb{Z}_{3Q}$ for $i \in \{1,2,3\}$ and ${\bf s}\in {\cal S}_{N+1}$. If we let $I=(N+1)^9$, $ Q  = \frac{1}{3}P^{(1-\varepsilon)/(I+2\varepsilon)}$, and $\Gamma = c_1 P^{(I-2+4\varepsilon)/(2(I+2\varepsilon))}$, for some positive constant $c_1$ and $\varepsilon$, we can show that the power constant at each transmitter is satisfied~\cite{mohammad-complex}. 
At time $t$, in each receiver $i$, we apply  Maximum Likelihood  (ML) detection to estimate $r_{i,\bf s} (t)$, for all  ${\bf s}\in {\cal S}_{N+1}$ from the received signal $y_i(t)$. Notice that the above detection could be erroneous  for  some receivers.
For now, let us focus on a specific time $t$, where $r_{i,\bf s} (t)$, ${\bf s}\in {\cal S}_{N+1}$ are correctly detected for all receivers and we will discuss the probability of error and its effect on the achievable rates later. 

In the rest of this proof we aim to show that the receivers one, two, and three can respectively resolve desired symbols
$a_{\bf s}$, $b_{\bf s}$ and $c_{\bf s}$  for all ${\bf s}\in {\cal S}_{N}$ based on the already individually detected sums $r_{i,\bf s}$ by successively exchanging and processing   information over the backhaul.  For convenience, in the notation to follow, we will denote addition in the sub-message 
vector 
indices ${\bf s}\in {\cal S}_{N}$ with corresponding superscripts; e.g, $a_{(s_{11}-1),s_{12},...,s_{33}}$ in (\ref{obs1}) will be 
written as $a_{s_{11}^{-1},s_{12},...,s_{33}}$. 

To start unraveling $r_{i,\bf s}$ for all ${\bf s} \in {\cal S}_{N+1}$,  \mbox{$i\in\{1,2,3\}$}, and resolve the  desired symbols at each receiver, we start from the boundaries as follows. 
Notice that for all \mbox{${\bf s'} \in \{{\bf s}\in {\cal S}_{N+1} : s_{31}=N+1\}$}, we have that $r_{3,\bf s'} = a_{s_{11},s_{12},s_{13},s_{21},s_{22},s_{23},N,s_{32},s_{33}},$ and hence receiver 3 has resolved some of  the symbols of receiver 1  without interference. Hence receiver 3 can give these symbols to receiver 1 directly over the backhaul.
The result of this message passing step is that receiver 1 knows its desired symbols $a_{\bf s}$, for all ${\bf s} \in  {\cal S}_{N}$ with $s_{31}=N$.    The corresponding backhaul rate that has been used for this step is given by $N^{8}\log(2\lfloor Q \rfloor +1)$.

Now receiver 1 can subtract $a_{s_{11},s_{12},s_{13},s_{21},s_{22},s_{23},N,s_{32},s_{33}}$ from its corresponding  observations in (\ref{obs1}) and obtain the  interference terms
\vspace{-0.1in} 
\begin{equation}b_{s_{11},s_{12},s_{13},s_{21},s_{22},s_{23},N,s_{32},s_{33}}+
c_{s_{11},s_{12}^{+1},s_{13}^{-1},s_{21},s_{22},s_{23},N,s_{32},s_{33}}, 
\label{rx1:2}\end{equation}
for all ${\bf s}\in {\cal S}_{N}$ with $s_{31}=N$.
In order to help receiver 2, receiver 1 will form the sums
\begin{align} &b_{s_{11},s_{12},s_{13},s_{21},s_{22},s_{23},N,s_{32},s_{33}}\nonumber\\+ 
&c_{s_{11},s_{12}^{+1},s_{13}^{-1},s_{21},s_{22},s_{23},N,s_{32},s_{33}} \nonumber\\+ &a_{s_{11},s_{12}^{+1},s_{13}^{-1},s_{21}^{-1},s_{22},s_{23}^{+1},N,s_{32},s_{33}}\;,
\label{sumsfrom1}\end{align}  
by adding the  symbols $a_{s_{11},s_{12}^{+1},s_{13}^{-1},s_{21}^{-1},s_{22},s_{23}^{+1},N,s_{32},s_{33}}$ to the interference terms in (\ref{rx1:2}), and give them to receiver 2 over the backhaul. 
From (\ref{obs2}) we can see that receiver 2 has already detected the sums\vspace{-0.1in} 
\begin{align} 
&a_{s_{11},s_{12}^{+1},s_{13}^{-1},s_{21}^{-1},s_{22},s_{23}^{+1},N,s_{32},s_{33}}\nonumber\\+ 
&b_{s_{11},s_{12}^{+1},s_{13}^{-1},s_{21},s_{22}^{-1},s_{23}^{+1},N,s_{32},s_{33}}\nonumber\\+
&c_{s_{11},s_{12}^{+1},s_{13}^{-1},s_{21},s_{22},s_{23},N,s_{32},s_{33}}
\end{align}
and hence subtracting them from (\ref{sumsfrom1}) will create the terms
\begin{equation}
b_{s_{11},s_{12},s_{13},s_{21},s_{22},s_{23},N,s_{32},s_{33}}\hspace{-0.02in} -  b_{s_{11},s_{12}^{+1},s_{13}^{-1},s_{21},s_{22}^{-1},s_{23}^{+1},N,s_{32},s_{33}},
\label{user2}
\end{equation}
from which receiver 2 can successively resolve its desired symbols $b_{\bf s}$ for all ${\bf s} \in  {\cal S}_{N}$ with $s_{31}=N$.\footnote{The corresponding recursion can be performed in $N$ iterations on the index $s_{23}$; letting $i=1,...,N$ and setting $s_{23}=N-i+1$, we can see that in the $i$th step, receiver 2 can successfully resolve  $b_{s_{11},s_{12},s_{13},s_{21},s_{22},N-i+1,N,s_{32},s_{33}}$, for all ${\bf s} \in  {\cal S}_{N}$ with $s_{23}=N-i+1$ and $s_{31}=N$. } The backhaul load  for this step is  equal to $N^{8}\log(6\lfloor Q \rfloor+1)$. 

Now we can already see that since receiver~2 knows $b_{s_{11},s_{12},s_{13},s_{21},s_{22},s_{23},N,s_{32},s_{33}}$, it can extract from (\ref{obs2}) the {interference} terms 
\vspace{-0.1in}
\begin{equation}
c_{s_{11},s_{12},s_{13},s_{21},s_{22},s_{23},N,s_{32},s_{33}}+ a_{s_{11},s_{12},s_{13},s_{21}^{-1},s_{22},s_{23}^{+1},N,s_{32},s_{33}},
\label{2to3N}
\end{equation}
for all ${\bf s}\in {\cal S}_{N}$ with $s_{31}=N$. Giving the above terms to receiver 3 over the backhaul, will help it to decode $c_{\bf s}$ for all ${\bf s} \in  {\cal S}_{N}$ with $s_{31}=N$, because  receiver 3 already has $a_{\bf s}$ for all ${\bf s} \in  {\cal S}_{N}$ with $s_{31}=N$. 
Since the number of symbols that have been exchanged in (\ref{2to3N}) 
is  $N^{8}$, the backhaul rate used for this step is equal to $N^{8}\log(4\lfloor Q \rfloor+1)$. 

Knowing $c_{s_{11},s_{12},s_{13},s_{21},s_{22},s_{23},N,s_{32},s_{33}}$, receiver 3 can extract from (\ref{obs3}) the interference terms 
\vspace{-0.05in}
\begin{equation}
a_{s_{11},s_{12},s_{13},s_{21},s_{22},s_{23},(N-1),s_{32},s_{33}}+ b_{s_{11},s_{12},s_{13},s_{21},s_{22},s_{23},N,s_{32}^{-1},s_{33}}
\label{known3}
\end{equation}
and subsequently create the terms
\begin{align}
&a_{s_{11},s_{12},s_{13},s_{21},s_{22},s_{23},(N-1),s_{32},s_{33}}\nonumber\\+ &b_{s_{11},s_{12},s_{13},s_{21},s_{22},s_{23},N,s_{32}^{-1},s_{33}} \nonumber\\+ 
&c_{s_{11},s_{12}^{+1},s_{13}^{-1},s_{21},s_{22},s_{23},N,s_{32}^{-1},s_{33}} 
\label{3to1}
\end{align}
by adding $c_{s_{11},s_{12}^{+1},s_{13}^{-1},s_{21},s_{22},s_{23},N,s_{32}^{-1},s_{33}}$
to (\ref{known3})
. Giving (\ref{3to1}) to receiver 1 will help it to resolve the symbols $a_{\bf s}$ for all ${\bf s} \in  {\cal S}_{N}$ with $s_{31}=N-1$, because  receiver 1 already knows  
the interfering sums $b_{s_{11},s_{12},s_{13},s_{21},s_{22},s_{23},N,s_{32}^{-1},s_{33}} + 
c_{s_{11},s_{12}^{+1},s_{13}^{-1},s_{21},s_{22},s_{23},N,s_{32}^{-1},s_{33}}$ from (\ref{rx1:2}). Similarly, as in (\ref{sumsfrom1}) and (\ref{user2}), receiver 1 will  help receiver 2 to resolve $b_{\bf s}$ for all ${\bf s} \in  {\cal S}_{N}$ with $s_{31}=N-1$. Now that 
receiver 2 knows $b_{s_{11},s_{12},s_{13},s_{21},s_{22},s_{23},(N-1),s_{32},s_{33}}$, it can extract from (\ref{obs2}) the interference terms 
\begin{equation*}
c_{s_{11},s_{12},s_{13},s_{21},s_{22},s_{23},(N-1),s_{32},s_{33}}+ a_{s_{11},s_{12},s_{13},s_{21}^{-1},s_{22},s_{23}^{+1},(N-1),s_{32},s_{33}}
\end{equation*}
and create the terms 
\begin{align}
&c_{s_{11},s_{12},s_{13},s_{21},s_{22},s_{23},(N-1),s_{32},s_{33}}\nonumber\\+ &a_{s_{11},s_{12},s_{13},s_{21}^{-1},s_{22},s_{23}^{+1},(N-1),s_{32},s_{33}}\nonumber\\+
&b_{s_{11},s_{12},s_{13},s_{21}^{-1},s_{22},s_{23}^{+1},N,s_{32}^{-1},s_{33}}
\label{2to3}
\end{align}
by adding the  symbols $b_{s_{11},s_{12},s_{13},s_{21}^{-1},s_{22},s_{23}^{+1},N,s_{32}^{-1},s_{33}}$ 
(that are already known to receiver~2) 
in order to match the interference terms in (\ref{known3}) that are known to receiver 3. Now from (\ref{2to3}) and (\ref{known3}) receiver 3 can extract $c_{\bf s}$ for all ${\bf s} \in  {\cal S}_{N}$ with $s_{31}=N-1$ and create as in (\ref{3to1}) the terms 
\begin{align}
&a_{s_{11},s_{12},s_{13},s_{21},s_{22},s_{23},(N-2),s_{32},s_{33}}\nonumber\\+ &b_{s_{11},s_{12},s_{13},s_{21},s_{22},s_{23},(N-1),s_{32}^{-1},s_{33}}\nonumber\\+
&c_{s_{11},s_{12}^{+1},s_{13}^{-1},s_{21},s_{22},s_{23},(N-1),s_{32}^{-1},s_{33}}
\end{align}
to help receiver 1 resolve $a_{\bf s}$ for all ${\bf s} \in  {\cal S}_{N}$ with $s_{31}=N-2$.

Following the same pattern we can see that  the backhaul communication will require $N$ rounds, where in each round $r\in\{0,...,N-1\}$, the receivers can resolve their desired symbols $a_{\bf s}$, $b_{\bf s}$ and $c_{\bf s}$ for all ${\bf s} \in  {\cal S}_{N}$ with $s_{31}=N-r$.  
In round $r$ we have the following message passing steps:
\begin{itemize}
\item Receiver 3 gives to receiver 1
\begin{align}
\hspace{-0.2in}{\cal M}_{3\rightarrow 1}^{[r]}\hspace{-0.05in} =\big\{
&b_{s_{11},s_{12},s_{13},s_{21},s_{22},s_{23},(N-r+1),s_{32}^{-1},s_{33}}\\+ &c_{s_{11},s_{12}^{+1},s_{13}^{-1},s_{21},s_{22},s_{23},(N-r+1),s_{32}^{-1},s_{33}}\nonumber\\ +
&a_{s_{11},s_{12},s_{13},s_{21},s_{22},s_{23},(N-r),s_{32},s_{33}}|\mbox{\footnotesize $s_{ij} \in \{1,...,N\}$}\nonumber
\big\},
\end{align}
\item Receiver 1 gives to receiver 2
\begin{align}
\hspace{-0.2in}{\cal M}_{1\rightarrow 2}^{[r]} \hspace{-0.05in}=\big\{&a_{s_{11},s_{12}^{+1},s_{13}^{-1},s_{21}^{-1},s_{22},s_{23}^{+1},(N-r),s_{32},s_{33}} \\+&c_{s_{11},s_{12}^{+1},s_{13}^{-1},s_{21},s_{22},s_{23},(N-r),s_{32},s_{33}}\nonumber\\+
&b_{s_{11},s_{12},s_{13},s_{21},s_{22},s_{23},(N-r),s_{32},s_{33}}|\mbox{\footnotesize $s_{ij} \in \{1,...,N\}$}\nonumber
\big\},
\end{align}
\item  Receiver 2 gives to receiver 3
\begin{align}
\hspace{-0.2in}{\cal M}_{2\rightarrow 3}^{[r]}\hspace{-0.05in} =\big\{&a_{s_{11},s_{12},s_{13},s_{21}^{-1},s_{22},s_{23}^{+1},(N-r),s_{32},s_{33}} \\+&b_{s_{11},s_{12},s_{13},s_{21}^{-1},s_{22},s_{23}^{+1},(N-r+1),s_{32}^{-1},s_{33}}\nonumber\\+&
c_{s_{11},s_{12},s_{13},s_{21},s_{22},s_{23},(N-r),s_{32},s_{33}}|\mbox{\footnotesize $s_{ij} \in \{1,...,N\}$}\nonumber
\big\}.
\end{align}
\end{itemize}

The above set of symbols ${\cal M}_{i\rightarrow j}^{[r]}$ are carefully created in each round $r$ based on the available observations and symbols 
at receiver $i$ in the previous rounds such that the interfering terms match exactly the interference terms observed  at receiver~$j$. The recipient of ${\cal M}_{i\rightarrow j}^{[r]}$ is therefore enabled to extract its desired symbol by subtracting the sum of the interfering symbols without knowing their individual values. 

The  total number of symbols that  have been exchanged over the backhaul in the above scheme is given by
\begin{align}
\#\textrm{msg} &=  \sum_{r=0}^{N-1} \left(\big|{\cal M}_{1\rightarrow 2}^{[r]}\big|+\big|{\cal M}_{2\rightarrow 3}^{[r]}\big| + \big|{\cal M}_{3\rightarrow 1}^{[r]}\big|\right)
\end{align}
where each symbol is in $\mathbb{Z}_{3Q}$. Therefore, the average (per user) backhaul rate that has been used can be calculated as  
$\overline R_{\rm b} \leq \frac{ \#\textrm{msg} \cdot \log(6Q+1)}{3}$. Since $\big|{\cal M}_{1\rightarrow 2}^{[r]}\big|=\big|{\cal M}_{2\rightarrow 3}^{[r]}\big| = \big|{\cal M}_{3\rightarrow 1}^{[r]}\big| = N^{8}$, for all $r=0,1,...,N-1$, we have that $\#\textrm{msg} = 3N^{9}$. Then 
$\lim_{P \rightarrow \infty} \frac{\overline R_{\rm b}}{\log(P)} \leq  N^9 \frac{1-\varepsilon}{(N+1)^9+2\varepsilon}$, which is arbitrary close to one, given large enough $N$ and small enough~$\varepsilon$. 

For each time slot $t$ such that the ML detection of $r_{i,\bf s} (t)$, \mbox{$\forall {\bf s}\in {\cal S}_{N+1}$} at all three receivers is performed without an error, the above unraveling process guarantees that receivers one, two, and three will be able to obtain respectively the correct  $a_{\bf s}(t)$, $ b_{\bf s}(t)$, and $c_{\bf s}(t)$, for all ${\bf s}\in {\cal S}_{N}$. On the other hand, for the time slots $t$ such that error occurs in the ML detection at {\em any of the receivers}, 
 the above unraveling process will most likely deliver incorrect symbols for some of the receivers. 
Notice however that error propagation occurs across the users, because of the unraveling process, but it is confined to 
symbols transmitted in the same time slot, i.e., symbol detection errors do not propagate across time. Such detection errors can be handled 
by using standard outer coding on each data stream. Using the union bound,  we can show that at each time $t$  
the probability of error in detecting $r_{i,\bf s} (t)$,  for some ${\bf s}\in {\cal S}_{N+1}$ and some $i\in\{1,2,3\}$, is upper-bounded 
by \mbox{$p_e \triangleq 3(N+1)^9 \exp(-c_2P^{\varepsilon/2 })$}, for some constant $c_2$~\cite{mohammad-complex}.  
Using Fano's inequality (see \cite[Eq. (14)]{mgmk09}),  the achievable rate of each 
sub-stream can be lower-bounded by $\left(1-p_e\right) \log(2 Q +1) -1$, yielding the DoF of $\frac{1-\varepsilon}{(N+1)^9+2\varepsilon}$. 
Therefore, the DoF of each message can be at least  $ N^9 \frac{1-\varepsilon}{(N+1)^9+2\varepsilon}$, which is arbitrary close to one, 
given large enough $N$ and small enough $\varepsilon$.

\section{Cooperation at the Transmitter Side}  \label{transmitter-cooperation-section}

In this section we are going to extend our results for an interference channel model in which cooperation is available between the \emph{transmitters} instead of receivers. More specifically, we will consider the same backhaul connectivity model as in the previous sections but with the role of transmitters and receivers being exchanged; that is, the transmitters will be allowed to cooperate by exchanging backhaul messages in an effort to jointly encode their transmitted signals to mitigate interference, while the receivers will attempt to decode their intended messages solely based on their received observations.  
 
 As a first example one can consider again the two-user interference channel with transmitter cooperation that has been studied in \cite{wt-tx11}.   
It can be shown (using the constant gap capacity approximation of \cite{wt-tx11}) that the optimal communication vs cooperation tradeoff 
for the transmitter cooperation case with two users is  again given by $\msf{DoF}^{*}(\alpha)=\min\{1,\frac{1+\alpha}{2}\}$, and hence it matches exactly the corresponding tradeoff that we have seen for the receiver cooperation; $\msf{DoF}(0)=1/2$ can  again be achieved by orthogonal user scheduling and  $\msf{DoF}(1)=1$ can be achieved the if the two transmitters exchange their messages $W_{i}$ and precode to zero-force interference.    
Our results on transmitter cooperation are given by the following theorems, which are pleasingly symmetric to Theorems \ref{thm:out} and \ref{thm:ach}, respectively: 
    
\begin{thm}[Upper Bound]
\label{thm:out2}
 In the $K$-user interference channel with transmitter cooperation and average backhaul load $\alpha$, we have that  \vspace{-0.05in}
\begin{equation}\msf{DoF}^*(\alpha) \leq \min\left\{1,\frac{1+\alpha}{2}\right\}.\end{equation}
\end{thm}

\begin{thm}[Achievability]  
\label{thm:ach2}
In the three-user interference channel with transmitter cooperation and average backhaul load $\alpha$,  we have that \vspace{-0.05in}
 \begin{equation}
 \msf{DoF}^*(\alpha) \geq \min\left\{1,\frac{1+\alpha}{2}\right\}.
 \end{equation}
\end{thm}

Obviously, the immediate consequence is that for the case of $K = 3$ users, the optimal communication vs cooperation tradeoff curves
for cooperation at the transmitters or at the receivers are identical, and given by $\msf{DoF}^*(\alpha) = \min\left\{1,\frac{1+\alpha}{2}\right\}$. 
 Theorem \ref{thm:out2} is proved in Appendix \ref{proof:conv2}, while Theorem \ref{thm:ach2} is proved in Section \ref{sec:proof4}, after formally defining 
the channel model for transmitter cooperation. 

\subsection{Cooperation at the Transmitters: Channel Model}

The channel model considered in this section is illustrated in Fig.~\ref{channelTx}. 
As in the case of receiver cooperation, we assume that
each transmitter $i\in \{1,...,K\}$ has a message $W_{i}$ that is intended for receiver $i$ and we restrict each codeword $[x_{i}(t)]_{t=1}^{n}$  to satisfy the  average power constraint $\frac{1}{n}\sum_{t=1}^{n}|x_{i}(t)|^{2}\leq P.$

The received signal at the $i$th receiver at time $t=1,...,n$ is given by
\begin{equation}
y_{i}(t) = \sum_{j=1}^{K} h_{ij} x_{j}(t) + z_{i}(t),\vspace{-0.05in}\end{equation}
where $h_{ij}\in \CC$ is the (complex) channel gain between the $j$th transmitter and the $i$th receiver, and $z_{i}(t)$ is the additive circularly-symmetric complex Gaussian noise observed at receiver $i$ with zero mean and unit variance.

\begin{figure}[ht]

                \centering
                \includegraphics[width=1\columnwidth]{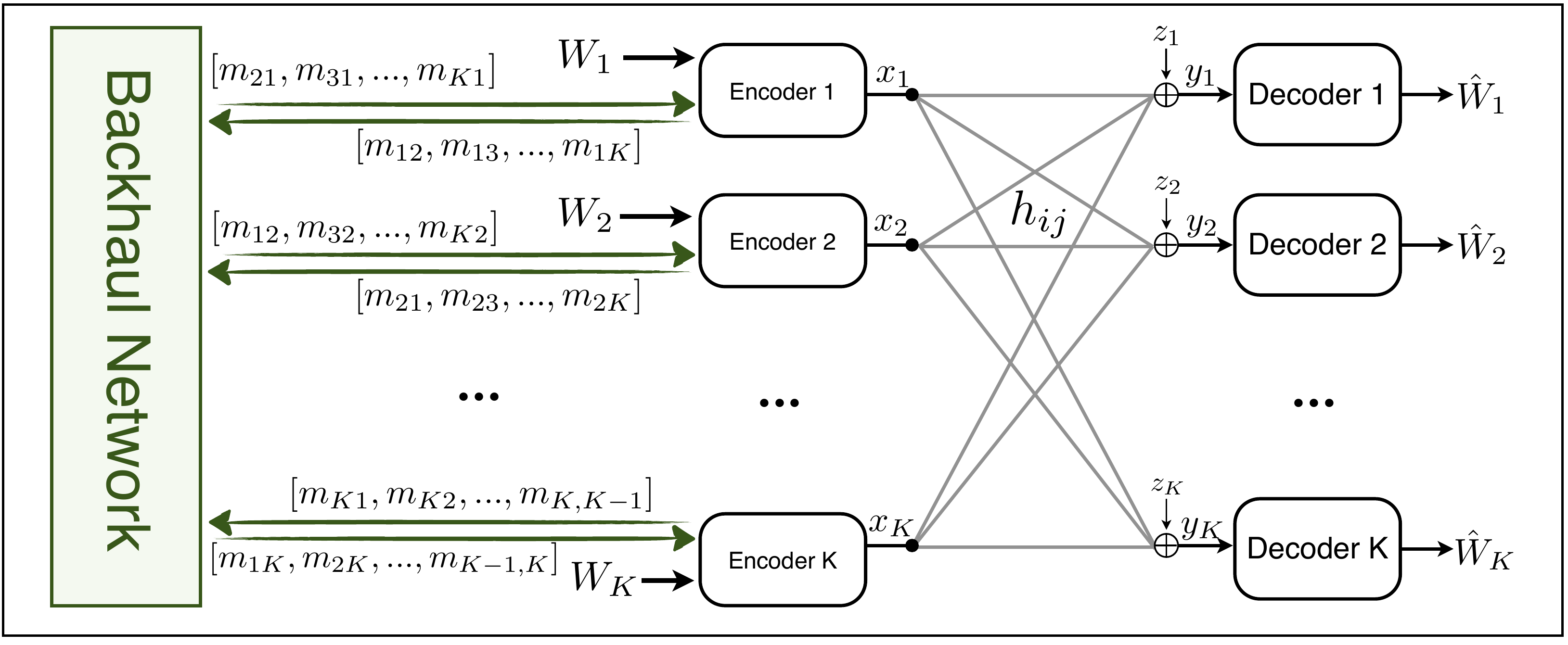}
                \caption{Channel model for transmitter cooperation}
                \label{channelTx}

\end{figure}

Further we assume that, in order to create $[x_{i}(t)]_{t=1}^{n}$, the encoders can collaborate over the backhaul by exchanging  messages  through directed noiseless links. We assume that there is a backhaul link $[i,\hat i]$, between every pair of encoders \mbox{$i\neq\hat i$~$\in \{1,...,K\}$}, and we denote the rate from encoder $i$ to  encoder $\hat i$ as $R_{\rm b}^{[i,\hat i]}$. The backhaul message from decoder $i$ to decoder $\hat i$, that passes through the link $[i,\hat i]$ at time $t$, is denoted by $m_{i\rightarrow \hat i}(t)$ and is given as a function of $W_{i}$ and all the previously received messages  $[m_{\ell\rightarrow  i}(\tau)]_{\tau=1}^{t-1}$ from  all encoders $\ell\in \{1,...,K\}, \ell\neq i$. 
 Hence, $R_{\rm b}^{[i,\hat i]}$ is again given by (\ref{ziopino}), and the total average backhaul rate
$\overline R_{\rm b}$ is defined as in (\ref{ziapina}).

The rate tuple  $(R_{1},R_{2},...,R_{K})$ is called \emph{achievable} under an average backhaul cooperation rate constraint $\overline R_{\rm b}\leq L$, 
if for any $\epsilon>0$ and sufficiently large $n$, there exist a length-$n$ coding scheme defined by:
\begin{itemize}
\item The message sets ${\cal W}_{i} = \{1,2,...,2^{nR_{i}}\}$, $i=1,...,K$.
\item The backhaul relaying functions $g_t^{[i,\hat i]}$ that generate $m_{i\rightarrow\hat i}(t)$ such that  
$$m_{i\rightarrow \hat i}(t) = g_t^{[i,\hat i]}\left(W_{i},M_{i}^{t-1}\right)\in {\cal B}^{[i,\hat i]},$$
where  $M_{i}^{t}\triangleq \left\{[m_{\ell\rightarrow  i}(\tau)]_{\tau=1}^{t} : \ell = 1,\ldots,K,\; \ell\neq i  \right\}$ 
is the collection of all the backhaul messages $m_{\ell\rightarrow i}(\tau)$ received at decoder $i$ up to time $t$, 
and ${\cal B}^{[i,\hat i]}$ is a finite set that denotes the  message alphabet used for the backhaul link $[i,\hat i]$.
\item The encoding functions 
$$ f_i : {\cal W}_{i} \times \prod_{\substack{\ell  = 1,...,K}{\ell \neq i}} \left ( {\cal B}^{[\ell,i]} \right )^n \rightarrow \CC^{n},\;
\mbox{that give }\; [x_{i}(t)]_{t=1}^{n} \triangleq f_i \left(W_{i},M_{i}^{n}\right),$$
 
\item The decoding functions $ \eta_i : \mathbb{C}^{n}  \rightarrow {\cal W}_{i}$,
that give $\hat W_{i}\triangleq \eta_i\left([y_{i}(\tau)]_{\tau=1}^{n}\right)$,
\end{itemize}
such that the corresponding probability of error given by $P_{e}^{(n)} \triangleq\mathbb{P}\left(\bigcup_{i = 1}^K  \{\hat W_{i}\neq W_{i}\} \right)$ is less~than~$\epsilon$, and the average backhaul cooperation rate satisfies the backhaul load constraint 
$$\overline R_{\rm b}=\frac{1}{K}\sum_{i=1}^{K}\sum_{\hat i\neq i}\frac{1}{n}H \bigl([m_{i\rightarrow \hat i}(\tau)]_{\tau=1}^{n}\bigr)\leq L.$$

Finally, the  definitions of the capacity region ${\mathcal C}_L$, the average  (per user) backhaul cooperation load $\alpha$ and the average (per user) achievable degrees of freedom $\msf{DoF}(\alpha)$  follow  directly from the corresponding definitions in Section~\ref{sec:modelRX}.

\subsection{Achievability (Proof of Theorem~4)}   \label{sec:proof4}

In order to prove Theorem~4 it suffices to show the achievability of the corner point ${\sf DoF}(1)=1$.
Here we will consider the same modulation scheme as in Section~\ref{proof4uplink} for the receiver cooperation;  
first, the transmitters are going to split their own messages  $W_{1}$, $W_{2}$ and $W_{3}$ into $N^{9}$ sub-messages, for some $N \in \mathbb{N}$ and then modulate each one of them over the integer constellation  $\mathbb{Z}_{Q}\triangleq\mathbb{Z}\cap [-Q,Q]$ to form the corresponding sub-streams. Following the notation introduced in Section~\ref{proof4uplink}, the sub-streams intended for user one, two and three, will be indexed by ${\bf s}\in {\cal S}_{N}$ and will be denoted by $\{a_{\bf s}(t)\}_{t=1}^n$, $\{b_{\bf s}(t)\}_{t=1}^n$ and $\{c_{\bf s}(t)\}_{t=1}^n$, respectively. For the rest of the proof we will drop the time index $t$ (unless it is required for clarification) and focus on a single time slot.

Notice that without any cooperation, after  using the above scheme, each transmitter will have access only to its own sub-streams $a_{\bf s}$, $b_{\bf s}$ and $c_{\bf s}$, respectively, for all ${\bf s}\in {\cal S}_{N}$.   
In the first part of the proof we aim to show that it is possible for the transmitters to cooperate over the backhaul (with average  backhaul rate   
$\overline R_{\rm b} \leq 1$) and create the combinations,  
\begin{align}
r_{1,\bf s}\;\triangleq  \;\;&a_{(s_{11}{-1}),s_{12},s_{13},s_{21},s_{22},s_{23},s_{31},s_{32},s_{33}} \nonumber\\
+&b_{s_{11},(s_{12}{-1}),s_{13},s_{21},s_{22},s_{23},s_{31},s_{32},s_{33}}\nonumber\\
+&c_{s_{11},s_{12},(s_{13}{-1}),s_{21},s_{22},s_{23},s_{31},s_{32},s_{33}},\nonumber
\end{align}
\begin{align}
r_{2,\bf s}\;\triangleq  \;\;&a_{s_{11},s_{12},s_{13},(s_{21}{-1}),s_{22},s_{23},s_{31},s_{32},s_{33}}\nonumber\\+
&b_{s_{11},s_{12},s_{13},s_{21},(s_{22}{-1}),s_{23},s_{31},s_{32},s_{33}}\nonumber\\+
&c_{s_{11},s_{12},s_{13},s_{21},s_{22},(s_{23}{-1}),s_{31},s_{32},s_{33}},\nonumber
\end{align}
and
\begin{align}
r_{3,\bf s}\;\triangleq  \;\;&a_{s_{11},s_{12},s_{13},s_{21},s_{22},s_{23},(s_{31}{-1}),s_{32},s_{33}}\nonumber\\+
&b_{s_{11},s_{12},s_{13},s_{21},s_{22},s_{23},s_{31},(s_{32}{-1}),s_{33}}\nonumber\\+
&c_{s_{11},s_{12},s_{13},s_{21},s_{22},s_{23},s_{31},s_{32},(s_{33}{-1})}, \nonumber
\end{align}
for all ${\bf s}\in {\cal S}_{N+1}$.
Recall that these combinations are the same with the ones defined in (\ref{obs1}), (\ref{obs2}) and (\ref{obs3}) in Section~\ref{proof4uplink}. 

First, notice that  for all ${\bf s}\in {\cal S}_{N+1}$ with $s_{21}=1$, $r_{2,\bf s}$  is given by 
$
b_{s_{11},s_{12},s_{13},1,s_{22}^{-1},s_{23},s_{31},s_{32},s_{33}} +
c_{s_{11},s_{12},s_{13},1,s_{22},s_{23}^{-1},s_{31},s_{32},s_{33}},$ and hence transmitter 2 only needs to obtain over the backhaul the corresponding $c_{\bf s}$ sub-streams  from transmitter 3.
Therefore in the first step, transmitter 3 will give $c_{s_{11},s_{12},s_{13},1,s_{22},s_{23}^{-1},s_{31},s_{32},s_{33}}$ to transmitter 2 over the backhaul, so that transmitter 2 can combine it with the its own  $b_{s_{11},s_{12},s_{13},1,s_{22}^{-1},s_{23},s_{31},s_{32},s_{33}}$ sub-streams  in order to create $r_{2,\bf s}$ for all ${\bf s}\in {\cal S}_{N+1}$ with $s_{21}=1$. Observe that in this first backhaul step transmitter 3 has sent the backhaul messages ${\cal M}_{3\rightarrow 2}^{[1]} = \{ c_{s_{11},s_{12},s_{13},1,s_{22},s_{23}^{-1},s_{31},s_{32},s_{33}} | \mbox{\footnotesize $s_{ij} \in \{1,...,N+1\}$} \}$ 
which can be equivalently written~as  \begin{equation}
{\cal M}_{3\rightarrow 2}^{[1]} = \{ c_{s_{11},s_{12},s_{13},1,s_{22},s_{23},s_{31},s_{32},s_{33}} | \mbox{\footnotesize $s_{ij} \in \{1,...,N\}$} \},
\end{equation}
since, by convention, $a_{\bf s} =b_{\bf s} =c_{\bf s} = 0$, for any ${\bf s}\notin{\cal S}_{N}$.

Recombining the newly acquired $c_{s_{11},s_{12},s_{13},1,s_{22},s_{23},s_{31},s_{32},s_{33}}$ sub-streams from transmitter 3, with the a different choice of its own sub-streams, transmitter 2 can create
the backhaul messages
\begin{align}
\hspace{-0.2in}{\cal M}_{2\rightarrow 1}^{[1]}\hspace{-0.05in} =\big\{
 &b_{s_{11},s_{12}^{-1},s_{13},1,s_{22},s_{23},s_{31},s_{32},s_{33}}\nonumber\\
+&c_{s_{11},s_{12},s_{13}^{-1},1,s_{22},s_{23},s_{31},s_{32},s_{33}}
|\mbox{\footnotesize $s_{ij} \in \{1,...,N+1\}$}
\big\},
\end{align}
and send it to transmitter 1 over the backhaul. Transmitter 1 can now add its own sub-streams $a_{s_{11}^{-1},s_{12},s_{13},1,s_{22},s_{23},s_{31},s_{32},s_{33}}$ to the newly acquired sub-streams 
\begin{equation}b_{s_{11},s_{12}^{-1},s_{13},1,s_{22},s_{23},s_{31},s_{32},s_{33}} + c_{s_{11},s_{12},s_{13}^{-1},1,s_{22},s_{23},s_{31},s_{32},s_{33}},\label{2to1d}\end{equation}
and create $r_{1,\bf s}$ for all ${\bf s}\in {\cal S}_{N+1}$ with $s_{21}=1$.

Now, based on this information, transmitter 1 can also create the backhaul messages
\begin{align}
\hspace{-0.2in}{\cal M}_{1\rightarrow 3}^{[1]}\hspace{-0.05in} =\big\{
 &a_{s_{11},s_{12}^{-1},s_{13},1,s_{22},s_{23},s_{31}^{-1},s_{32}^{+1},s_{33}}\nonumber\\
+&b_{s_{11},s_{12}^{-1},s_{13},1,s_{22},s_{23},s_{31},s_{32},s_{33}}\nonumber\\
+&c_{s_{11},s_{12},s_{13}^{-1},1,s_{22},s_{23},s_{31},s_{32},s_{33}}
|\mbox{\footnotesize $s_{ij} \in \{1,...,N+1\}$}
\big\},\label{1to3b}
\end{align}
by adding $a_{s_{11},s_{12}^{-1},s_{13},1,s_{22},s_{23},s_{31}^{-1},s_{32}^{+1},s_{33}}$ to the sub-streams that have been obtained from transmitter~2 in (\ref{2to1d}). Giving ${\cal M}_{1\rightarrow 3}^{[1]}$ to transmitter 3 over the backhaul is able to help transmitter 3 create 
\begin{align}
\hspace{-0.2in}
&a_{s_{11},s_{12}^{-1},s_{13},1,s_{22},s_{23},s_{31}^{-1},s_{32}^{+1},s_{33}}\nonumber\\
+&b_{s_{11},s_{12}^{-1},s_{13},1,s_{22},s_{23},s_{31},s_{32},s_{33}}\nonumber\\
+&c_{s_{11},s_{12}^{-1},s_{13},1,s_{22},s_{23},s_{31},s_{32}^{+1},s_{33}^{-1}},\label{trianta}
\end{align}
for all $s_{ij} \in \{1,...,N+1\}$ by adding $c_{s_{11},s_{12}^{-1},s_{13},1,s_{22},s_{23},s_{31},s_{32}^{+1},s_{33}^{-1}} - c_{s_{11},s_{12},s_{13}^{-1},1,s_{22},s_{23},s_{31},s_{32},s_{33}}$ to the sub-streams in (\ref{1to3b}). Notice that (\ref{trianta}) can be equivalently written as  
\begin{align}
\hspace{-0.2in}
&a_{s_{11},s_{12},s_{13},1,s_{22},s_{23},s_{31}^{-1},s_{32},s_{33}}\nonumber\\
+&b_{s_{11},s_{12},s_{13},1,s_{22},s_{23},s_{31},s_{32}^{-1},s_{33}}\nonumber\\
+&c_{s_{11},s_{12},s_{13},1,s_{22},s_{23},s_{31},s_{32},s_{33}^{-1}},\label{triantaena}
\end{align}
and hence we can see that transmitter 3 can also create the required combinations $r_{3,\bf s}$ for all ${\bf s}\in {\cal S}_{N+1}$ with $s_{21}=1$.

For the next step, transmitter 3 can substitute $c_{s_{11},s_{12},s_{13},1,s_{22},s_{23},s_{31},s_{32},s_{33}^{-1}}$ with $c_{s_{11},s_{12},s_{13},2,s_{22},s_{23}^{-1},s_{31}^{-1},s_{32},s_{33}}$ in (\ref{triantaena}) in order to create the backhaul messages
\begin{align}
\hspace{-0.2in}{\cal M}_{3\rightarrow 2}^{[2]}\hspace{-0.05in} =\big\{
&a_{s_{11},s_{12},s_{13},1,s_{22},s_{23},s_{31}^{-1},s_{32},s_{33}}\nonumber\\
+&b_{s_{11},s_{12},s_{13},1,s_{22},s_{23},s_{31},s_{32}^{-1},s_{33}}\nonumber\\
+&c_{s_{11},s_{12},s_{13},2,s_{22},s_{23}^{-1},s_{31}^{-1},s_{32},s_{33}}
|\mbox{\footnotesize $s_{ij} \in \{1,...,N+1\}$}
\big\},
\end{align}
and give it to transmitter 2 over the backhaul. After obtaining ${\cal M}_{3\rightarrow 2}^{[2]}$, transmitter 2 can create 
\begin{align}
\hspace{-0.2in}
&a_{s_{11},s_{12},s_{13},1,s_{22},s_{23},s_{31}^{-1},s_{32},s_{33}}\nonumber\\
+&b_{s_{11},s_{12},s_{13},2,s_{22}^{-1},s_{23},s_{31}^{-1},s_{32},s_{33}}\nonumber\\
+&c_{s_{11},s_{12},s_{13},2,s_{22},s_{23}^{-1},s_{31}^{-1},s_{32},s_{33}} \label{triantatria}
\end{align}
for all $s_{ij} \in \{1,...,N+1\}$, simply by adding $b_{s_{11},s_{12},s_{13},2,s_{22}^{-1},s_{23},s_{31}^{-1},s_{32},s_{33}} - b_{s_{11},s_{12},s_{13},1,s_{22},s_{23},s_{31},s_{32}^{-1},s_{33}}$. Notice that (\ref{triantatria}) can  equivalently be written as
\begin{align}
\hspace{-0.2in}
&a_{s_{11},s_{12},s_{13},1,s_{22},s_{23},s_{31},s_{32},s_{33}}\nonumber\\
+&b_{s_{11},s_{12},s_{13},2,s_{22}^{-1},s_{23},s_{31},s_{32},s_{33}}\nonumber\\
+&c_{s_{11},s_{12},s_{13},2,s_{22},s_{23}^{-1},s_{31},s_{32},s_{33}} 
\end{align} 
which is equal to $r_{2,\bf s}$ for all ${\bf s}\in {\cal S}_{N+1}$ with $s_{21}=2$. Based on this information, now transmitter~2 can create the backhaul messages  
\begin{align}
\hspace{-0.2in}{\cal M}_{2\rightarrow 1}^{[2]}\hspace{-0.05in} =\big\{
&a_{s_{11},s_{12},s_{13},1,s_{22},s_{23}^{+1},s_{31},s_{32},s_{33}}\nonumber\\
+&b_{s_{11},s_{12}^{-1},s_{13},2,s_{22},s_{23},s_{31},s_{32},s_{33}}\nonumber\\
+&c_{s_{11},s_{12},s_{13}^{-1},2,s_{22},s_{23},s_{31},s_{32},s_{33}}
|\mbox{\footnotesize $s_{ij} \in \{1,...,N+1\}$}
\big\},
\end{align}
that will allow transmitter~1 to create the combinations 
\begin{align}
\hspace{-0.2in}
&a_{s_{11}^{-1},s_{12},s_{13},2,s_{22},s_{23},s_{31},s_{32},s_{33}}\nonumber\\
+&b_{s_{11},s_{12}^{-1},s_{13},2,s_{22},s_{23},s_{31},s_{32},s_{33}}\nonumber\\
+&c_{s_{11},s_{12},s_{13}^{-1},2,s_{22},s_{23},s_{31},s_{32},s_{33}},
\end{align}
for all $s_{ij} \in \{1,...,N+1\}$ and hence also acquire  
$r_{1,\bf s}$ for all ${\bf s}\in {\cal S}_{N+1}$ with $s_{21}=2$.

Following the same pattern, before each round $r$, the transmitters 1, 2 and 3 will have available the combinations $r_{1,\bf s}$, $r_{2,\bf s}$ and  $r_{3,\bf s}$ for all ${\bf s}\in {\cal S}_{N+1}$ with $s_{21}=r-1$, respectively,  and will create the backhaul messages
\begin{align}
\hspace{-0.2in}{\cal M}_{3\rightarrow 2}^{[r]}\hspace{-0.05in} =\big\{
&a_{s_{11},s_{12},s_{13},(r-1),s_{22},s_{23},s_{31}^{-1},s_{32},s_{33}}\nonumber\\
+&b_{s_{11},s_{12},s_{13},(r-1),s_{22},s_{23},s_{31},s_{32}^{-1},s_{33}}\nonumber\\
+&c_{s_{11},s_{12},s_{13},r,s_{22},s_{23}^{-1},s_{31}^{-1},s_{32},s_{33}}
|\mbox{\footnotesize $s_{ij} \in \{1,...,N+1\}$}
\big\},
\end{align}

\begin{align}
\hspace{-0.2in}{\cal M}_{2\rightarrow 1}^{[r]}\hspace{-0.05in} =\big\{
&a_{s_{11},s_{12},s_{13},(r-1),s_{22},s_{23}^{+1},s_{31},s_{32},s_{33}}\nonumber\\
+&b_{s_{11},s_{12}^{-1},s_{13},r,s_{22},s_{23},s_{31},s_{32},s_{33}}\nonumber\\
+&c_{s_{11},s_{12},s_{13}^{-1},r,s_{22},s_{23},s_{31},s_{32},s_{33}}
|\mbox{\footnotesize $s_{ij} \in \{1,...,N+1\}$}
\big\},
\end{align}
and 
\begin{align}
\hspace{-0.2in}{\cal M}_{1\rightarrow 3}^{[r]}\hspace{-0.05in} =\big\{
 &a_{s_{11},s_{12}^{-1},s_{13},r,s_{22},s_{23},s_{31}^{-1},s_{32}^{+1},s_{33}}\nonumber\\
+&b_{s_{11},s_{12}^{-1},s_{13},r,s_{22},s_{23},s_{31},s_{32},s_{33}}\nonumber\\
+&c_{s_{11},s_{12},s_{13}^{-1},r,s_{22},s_{23},s_{31},s_{32},s_{33}}
|\mbox{\footnotesize $s_{ij} \in \{1,...,N+1\}$}
\big\},
\end{align}
in order to obtain $r_{1,\bf s}$, $r_{2,\bf s}$ and  $r_{3,\bf s}$ for all ${\bf s}\in {\cal S}_{N+1}$ with $s_{21}=r$. Therefore, after $N+1$ rounds of backhaul cooperation all the transmitters will have obtained $r_{1,\bf s}$, $r_{2,\bf s}$ and  $r_{3,\bf s}$ for all ${\bf s}\in {\cal S}_{N+1}$ as required.
In order to measure the amount of backhaul cooperation that is needed for the above message passing scheme we can argue that the total number of messages that have been exchanged is bounded by
\begin{align}
\#\textrm{msg} &=  \sum_{r=1}^{N+1} \left(\big|{\cal M}_{1\rightarrow 2}^{[r]}\big|+\big|{\cal M}_{2\rightarrow 3}^{[r]}\big| + \big|{\cal M}_{3\rightarrow 1}^{[r]}\big|\right)\leq 3(N+1)^{9},
\end{align}
and since each symbol is in $\mathbb{Z}_{3Q}$, the average (per user) backhaul rate that has been used can be calculated as  
$\overline R_{\rm b} \leq \frac{ \#\textrm{msg} \cdot \log(6\lfloor Q \rfloor+1)}{3}\leq (N+1)^{9}\log(6\lfloor Q \rfloor+1)$. 
Hence, choosing the same parameter $Q$ as in  Section~\ref{proof4uplink}, we have that  
$\lim_{P \rightarrow \infty} \frac{\overline R_{\rm b}}{\log(P)} \leq  \frac{(1-\varepsilon)(N+1)^9 }{(N+1)^9+2\varepsilon}$, which is arbitrary close to one, given large enough $N$ and small enough~$\varepsilon$. 

After the above cooperation alignment step over the backhaul, for each time slot $t$, each transmitter will form its transmitted signals from $r_{1,\bf s}(t)$, $r_{2,\bf s}(t)$ and  $r_{3,\bf s}(t)$, respectively, following a scheme similar to the aligned network diagonalization scheme that  has been 
introduced in \cite{sa14}. The main idea  is that the transmitters  will effectively diagonalize the channel matrix by choosing to transmit  specific  linear combinations of their own $r_{i,\bf s}(t)$, for all ${\bf s}\in {\cal S}_{N+1}$, with monomial coefficients 
generated from the inverse of the channel matrix.

Let us first define
$$\Hm \triangleq \begin{bmatrix}
 h_{11} &  h_{12}&  h_{13}\\
 h_{21} &  h_{22}&  h_{23}\\
 h_{31} &  h_{32}&  h_{33} 
\end{bmatrix}\;\;\;\;
\mbox{and}\;\; \;\;
\begin{bmatrix}
\hat h_{11} & \hat h_{12}& \hat h_{13}\\
\hat h_{21} & \hat h_{22}& \hat h_{23}\\
\hat h_{31} & \hat h_{32}& \hat h_{33}
\end{bmatrix}\triangleq \Hm^{-1}$$
and let $\hat \nu_{\bf s}= \prod_{i,j} \hat h_{ij}^{s_{ij}}$.
The transmitted signal from transmitter $i=1,2,3$ at time $t$ is given by 
\begin{align}
x_{i}(t) &= \Gamma'\cdot \sum_{{\bf s}\in {\cal S}_{N+1}} \hat\nu_{\bf s} \cdot
r_{i,\bf s}(t)  ,  
\label{rewriter}
\end{align}
where  the scaling factor $\Gamma'$ guarantees that the power constraint is satisfied. 
The corresponding received signals at time $t$ are given by 
\begin{equation}y_{i}(t) = h_{i1}x_{1}(t) + h_{i2}x_{2}(t)+h_{i3}x_{3}(t) + z_{i}(t),\; \;i=1,2,3.\end{equation}

If we rewrite  (\ref{rewriter})   as a sum over all ${\bf s}\in {\cal S}_{N}$ instead of ${\cal S}_{N+1}$ (by factoring out the corresponding channel coefficients) as 
\begin{align}
x_{i}(t) = \Gamma'\cdot \sum_{{\bf s}\in {\cal S}_{N}} \hat\nu_{\bf s} \cdot
\left(\hat h_{i1}a_{\bf s}(t)+\hat h_{i2}b_{\bf s}(t)+\hat h_{i3}c_{\bf s}(t)\right),
\label{sarantatrio}
\end{align}
and further define
\begin{align}
\xv(t) &\triangleq [x_{1}(t) , x_{2}(t), x_{3}(t)]^{\rm T},\\
\yv(t) &\triangleq [y_{1}(t) , y_{2}(t), y_{3}(t)]^{\rm T},\\
\zv(t) &\triangleq [z_{1}(t) , z_{2}(t), z_{3}(t)]^{\rm T},\\
\uv(t) &\triangleq [a_{\bf s}(t) , b_{\bf s}(t), c_{\bf s}(t)]^{\rm T},
\end{align}
we can rewrite the transmitted signals at time $t$ in matrix form as 
\begin{align}
\xv(t) &= \Gamma'\cdot \sum_{{\bf s}\in {\cal S}_{N}} \hat\nu_{\bf s} \cdot \Hm^{-1} \cdot \uv(t),
\end{align}
and the corresponding received signal observations as
\begin{align}
\yv(t) &= \Hm \cdot \xv(t) + \zv(t)\\
&= \Hm \cdot \left(\Gamma'\cdot \sum_{{\bf s}\in {\cal S}_{N}} \hat\nu_{\bf s} \cdot \Hm^{-1} \cdot \uv(t)\right) + \zv(t)\\
&= \Gamma'\cdot \sum_{{\bf s}\in {\cal S}_{N}} \hat\nu_{\bf s} \cdot \uv(t) + \zv(t).\label{rxsig}
\end{align}

Therefore, as we  can directly see from (\ref{rxsig}), the received signal observations at users one, two and three, are given by 
\begin{equation}
y_{1}(t) = \Gamma'\cdot \sum_{{\bf s}\in {\cal S}_{N}} \hat\nu_{\bf s} \cdot a_{\bf s}(t) + z_{1}(t),
\end{equation}  
\begin{equation}
y_{2}(t) = \Gamma'\cdot \sum_{{\bf s}\in {\cal S}_{N}} \hat\nu_{\bf s} \cdot b_{\bf s}(t) + z_{2}(t),
\end{equation} 
and
\begin{equation}
y_{3}(t) = \Gamma'\cdot \sum_{{\bf s}\in {\cal S}_{N}} \hat\nu_{\bf s} \cdot c_{\bf s}(t) + z_{3}(t),
\end{equation} 
respectively, and hence, the above transmission scheme is able to  effectively inverting the channel matrix and eliminate all interference; each receiver will only observe the sub-streams that correspond to its own desired message. 

To finalize the proof, following the same ML detection scheme that we have used for the receiver cooperation 
in Section~\ref{proof4uplink} and the corresponding results in \cite{sa14}, we can show that for each time~$t$,  all the sub-streams $a_{\bf s}(t)$, $b_{\bf s}(t)$ and $c_{\bf s}(t)$ can be successfully detected at their intended receivers, with the corresponding message DoF being arbitrarily close to one for large enough $N$.\footnote{Also in this case, to go from 
vanishing symbol-error probability to block-error probability the approach of per-stream outer coding and the argument based on Fano inequality as in 
\cite[Eq. (14)]{mgmk09} can be used.}

\section{Conclusions}\label{conclusions}

In  multiuser interference channel models, cooperation -- either at the receivers' or at the transmitters' side -- 
is often considered to be available in the network through some form of centralized processing.  This modeling approach is rather appealing, 
as it allows the entire cooperative network to be seen as a single MIMO multiple access or broadcast channel. 
Hence, a common theme in cooperative interference management techniques has been, up to now,  the use the backhaul links 
primarily as a means to offload baseband processing to a single central node. However,  this barely takes into account the inherent distributed nature of such systems; 
cooperative networks often consist of several distributed processors (e.g., receivers connected to the cloud) and {\em inter-processor communication} is a precious and 
limited resource that should also be quantified. 

Motivated by this consideration, our first goal in this paper has been to challenge the above current centralized approach and consider 
the extension to more than two users of the classical information-theoretic model of interference channels with receiver or transmitter cooperation pioneered in 
\cite{wt-Rx11,wt-tx11}. We considered a multiuser interference network under a general cooperation model that does not impose any specific, {\em a~priori} structure in the backhaul  
architecture. When receiver cooperation is available for example, every receiver can first process its observations locally and then potentially share information with any other receiver 
in the network in order to help in the decoding process. Of course, a centralized approach can be implemented within this framework as a special case, 
by restricting all receivers to just quantize and forward their observations to a single node within the network. Overall, the ``interactive'' approach adopted   
in this paper allows to consider a more general class of interference management strategies. 

For this model, we were able to quantify the fundamental tradeoff between the achievable communication rates and the corresponding backhaul 
cooperation rate in wireless networks in terms of degrees of freedom.  In particular, 
we showed that if the average (per user) rate scales as $R = \msf{DoF}\cdot\log(P) + o(\log(P))$ 
and the average (per user) backhaul cooperation load scales as $L = \alpha\cdot\log(P)+ o(\log(P))$,  the optimal  communication vs cooperation tradeoff for the $K$-user interference channel 
is upper bounded by $\msf{DoF}^{*}(\alpha)\leq\min\{1,\frac{1+\alpha}{2}\}$, regardless whether cooperation is available at the transmitters' or at receivers' side. 
The corner points of this tradeoff region, namely  $\msf{DoF}(0) =\frac{1}{2}$ and  $\msf{DoF}(1) ={1}$, are easily achievable in the case of $K=2$ by orthogonal user 
scheduling and centralized processing, respectively, and the optimal tradeoff in this case is given by $\msf{DoF}^{*}(\alpha)=\min\{1,\frac{1+\alpha}{2}\}$. 
However, following the same approach for $K>2$ is not optimal in general. 
This is expected -- at least for the case where there is no cooperation -- since  interference alignment can still achieve the corner point $\msf{DoF}(0) =\frac{1}{2}$, no matter how many interfering users 
are in the network. On the other hand, as the number of users increases, centralized processing can only achieve $\msf{DoF}(2(K-1)/K) ={1}$, and hence it seems that more backhaul capacity 
is required to maintain the full DoF.  

Surprisingly, this paper shows that this is not true. We developed the new idea of {\em cooperation alignment} and showed that it can have an analogous effect on the backhaul cooperation 
load as interference alignment has on the ``wireless'' degrees of freedom. That is, under cooperation alignment, from each receiver's (resp. transmitter's) perspective it appears as if a single 
user jointly processes the observations (resp. messages) of the entire network and only shares the necessary, minimal information over the backhaul. 
Focusing on the $K = 3$ users interference channel case, we proposed a new interference management scheme based on cooperation alignment and proved that it is able to 
achieve the corner point $\msf{DoF}(1) ={1}$, in  both the receiver cooperation and transmitter cooperation cases. This implies that, for $K=3$, cooperation alignment over the backhaul, together with interference alignment over the wireless channel, can  achieve the entire cooperation vs cooperation tradeoff, $\msf{DoF}^{*}(\alpha)=\min\{1,\frac{1+\alpha}{2}\}$,  which surprisingly remains unchanged 
as we move from two to three users. 

An interesting open question that arises from this work is whether this behavior continues to hold for cooperative interference networks with more than three transmit-receive pairs. 
For example, when $K=4$, centralized processing can achieve $\msf{DoF}(3/2) =1$, but it is not known whether the same full DoF can be achieved with an average per user backhaul 
load $\alpha< {3}/{2}$. More importantly, proving the achievability of  the corner point $\msf{DoF}(1) =1$ for networks with $K\geq4$, combined with the upper bounds presented 
in this work, would immediately yield a characterization of the optimal cooperation vs cooperation tradeoff in these cases. Even though such a generalization of our schemes can be very 
challenging -- mainly due to the constructive nature of our achievability proofs -- we believe that techniques based on cooperation alignment will eventually be able to break the ``centralized processing'' 
barrier in cooperative interference networks and provide a better understanding of the design of more efficient interference management schemes.

\clearpage
\appendices
\section{Proof of Theorem~\ref{thm:out}: Receiver Cooperation Upper Bound }\label{proof:conv}

First we are going to bound the rates $R_{1}+R_{2}$ by following an approach similar to the two-user bound developed in \cite{wt-Rx11}.
Consider a genie that gives   
$$y_{[2:K]}^{n}\triangleq [y_{2}(\tau),y_{3}(\tau),...,y_{K}(\tau)]_{\tau=1}^{n}\; \mbox{and}\;x_{2}^{n}\triangleq [x_{2}(\tau)]_{\tau=1}^{n}$$ to receiver one,  and the messages $$W_{[3:K]}\triangleq [W_{3},W_{4},...,W_{K}]$$ to both receivers one and two, as side information.  
Starting from Fano's inequality we have that

\resizebox{1\linewidth}{!}{
 \begin{minipage}{\linewidth}
\begin{align}
&n(R_{1}+R_{2}-\epsilon_{n}) \nonumber
\\&\stackrel{\phantom{\rm (0)}}{=} I\big(W_{1};y_{1}^{n},M_{1}^{[n]}\big) + 
I\big(W_{2};y_{2}^{n},M_{2}^{[n]}\big)\nonumber
\\&\stackrel{{\rm (a)}}{\leq} I\big(W_{1};y_{1}^{n},M_{1}^{[n]}\big|W_{[3:K]}\big) + 
I\big(W_{2};y_{2}^{n},M_{2}^{[n]}\big|W_{[3:K]}\big)\nonumber
\\&\stackrel{{\rm (b)}}{\leq} I\big(x_{1}^{n};y_{1}^{n},M_{1}^{[n]}\big|W_{[3:K]}\big) + 
I\big(x_{2}^{n};y_{2}^{n},M_{2}^{[n]}\big|W_{[3:K]}\big)\nonumber
\\&\stackrel{{\rm (c)}}{\leq} I\big(x_{1}^{n};y_{1}^{n},M_{1}^{[n]},x_{2}^{n},y_{[2:K]}^{n}\big|W_{[3:K]}\big) + 
I\big(x_{2}^{n};y_{2}^{n},M_{2}^{[n]}\big|W_{[3:K]}\big)\nonumber
\\&\stackrel{{\rm (d)}}{=} I\big(x_{1}^{n};y_{[1:K]}^{n}\big|x_{2}^{n},W_{[3:K]}\big) + 
I\big(x_{2}^{n};y_{2}^{n},M_{2}^{[n]}\big|W_{[3:K]}\big),
\label{firstp}\\\nonumber
\vspace{0.1in}
\end{align}
\end{minipage}}
where ($\rm a$) follows from the fact that $W_{i}$ are independent, ($\rm b$) from the data processing inequality, ($\rm c$) from the chain rule by adding $I\big(x_{1}^{n};x_{2}^{n},y_{[2:K]}^{n}\big|y_{1}^{n},M_{1}^{[n]},W_{[3:K]}\big)\geq 0$, and ($\rm d$) from the fact that $x_{1}^{n}$ and $x_{2}^{n}$ are independent and that $M_{1}^{[n]}$ is a function of $y_{[1:K]}^{n}$.

The second term in (\ref{firstp}) can be further bounded as 
\begin{align}
&I\big(x_{2}^{n};y_{2}^{n},M_{2}^{[n]}\big|W_{[3:K]}\big) \nonumber
\\ &\stackrel{\phantom{\rm (0)}}{=} I\big(x_{2}^{n};y_{2}^{n}\big|W_{[3:K]}\big) + I\big(x_{2}^{n};M_{2}^{[n]}\big|y_{2}^{n},W_{[3:K]}\big)\nonumber
\\ &\stackrel{\phantom{\rm (0)}}{\leq} I\big(x_{2}^{n};y_{2}^{n}\big|W_{[3:K]}\big) + H(M_{2}^{[n]})\nonumber
\\ &\stackrel{{\rm (e)}}{\leq} I\big(x_{2}^{n};y_{2}^{n}\big|W_{[3:K]}\big) 
+ \textstyle\sum_{i\neq2}H([m_{i\rightarrow2}(\tau)]_{\tau=1}^{n})\nonumber
\\ &\stackrel{\phantom{\rm (0)}}{=} I\big(x_{2}^{n};y_{2}^{n}\big|W_{[3:K]}\big) 
+ n\textstyle\sum_{i\neq2}R_{\rm b}^{[i,2]},\label{eq26}
\end{align}
where ($\rm e$) follows from the definition of $M_{2}^{[n]}\triangleq \left\{[m_{i\rightarrow  2}(\tau)]_{\tau=1}^{n}, i\in\{1,...,K\},i\neq 2  \right\}$, the chain rule and the fact that conditioning reduces entropy. Hence, 
\begin{equation}
n(R_{1}+R_{2}-\epsilon_{n}) \leq  I\big(x_{1}^{n};y_{[1:K]}^{n}\big|x_{2}^{n},W_{[3:K]}\big) + I\big(x_{2}^{n};y_{2}^{n}\big|W_{[3:K]}\big) 
+ n\textstyle\sum_{i\neq2}R_{\rm b}^{[i,2]}.
\label{eq:hence}
\end{equation}

%

The remaining two terms in the RHS of (\ref{eq:hence}) can be further bounded as

\resizebox{1\linewidth}{!}{
\hspace{-0.23in}
  \begin{minipage}{\linewidth}
\begin{align}
&I\big(x_{1}^{n};y_{[1:K]}^{n}\big|x_{2}^{n},W_{[3:K]}\big) + I\big(x_{2}^{n};y_{2}^{n}\big|W_{[3:K]}\big) \nonumber
\\&\stackrel{\phantom{\rm (a')}}{=} h\big(y_{1}^{n},y_{[3:K]}^{n}\big|y_{2}^{n},x_{2}^{n},W_{[3:K]}\big) -h\big(z_{[1:K]}^{n}\big)  + h\big(y_{2}^{n}\big|W_{[3:K]}\big) \nonumber
\\&\stackrel{\phantom{\rm (a')}}{=}h\big(h_{11} x_{1}^{n} +  z_{1}^{n},h_{31} x_{1}^{n} +  z_{3}^{n},...,h_{K1} x_{1}^{n} +  z_{K}^{n} \big| h_{21} x_{1}^{n} +  z_{2}^{n}\big)\nonumber 
\\&\hspace{0.3in}{+}h\big(h_{21} x_{1}^{n} +h_{22} x_{2}^{n}+  z_{2}^{n}\big)-h\big(z_{[1:K]}^{n}\big)\nonumber
\\&\stackrel{{\rm (a')}}{\leq} \sum_{i\neq2} h\big(h_{i1} x_{1}^{n} +  z_{i}^{n} \big|h_{21} x_{1}^{n} +  z_{2}^{n}\big)\nonumber 
\\&\hspace{0.5in}{+}h\big(h_{21} x_{1}^{n} +h_{22} x_{2}^{n}+  z_{2}^{n}\big)-h\big(z_{[1:K]}^{n}\big)\nonumber
\\&\stackrel{{\rm (b')}}{=} \sum_{i\neq2} h\big(h_{i1} x_{1}^{n} +  z_{i}^{n} -h_{i1}h_{21}^{-1}(h_{21} x_{1}^{n} +  z_{2}^{n}) \big|h_{21} x_{1}^{n} +  z_{2}^{n}\big)\nonumber 
\\&\hspace{0.5in}{+}h\big(h_{21} x_{1}^{n} +h_{22} x_{2}^{n}+  z_{2}^{n}\big)-h\big(z_{[1:K]}^{n}\big)\nonumber
\\&\stackrel{\phantom{\rm (a')}}{\leq} \sum_{i\neq2} h\big(z_{i}^{n} -h_{i1}h_{21}^{-1}z_{2}^{n}\big){+}h\big(h_{21} x_{1}^{n} +h_{22} x_{2}^{n}+  z_{2}^{n}\big)-h\big(z_{[1:K]}^{n}\big)\nonumber
\\&\stackrel{{\rm (c')}}{\leq} n\sum_{i\neq2} \log(1+|h_{i1}|^{2}/|h_{21}|^{2})
+ n \log(1+P(|h_{21}|^{2}+ |h_{22}|^{2})),\nonumber\\\nonumber
\end{align}
\end{minipage}}
where ($\rm a'$) follows from the chain rule and the fact that conditioning reduces entropy, ($\rm b'$) follows from the translational invariance property, 
and ($\rm c'$) from the fact that the Gaussian distribution maximizes entropy for a given variance. Putting everything together, we conclude that
\begin{equation}
\vspace{-0.01in}
R_{1}+R_{2}\leq \log(1+P(|h_{21}|^{2}+ |h_{22}|^{2})) + \sum_{i\neq2}R_{\rm b}^{[i,2]}
+o(\log(P)).
\end{equation}

In a similar way, we can obtain bounds of the same form for the pairs $R_{2}+R_{3}$, $R_{3}+R_{4}$, up to $R_{K-1}+R_{K}$ and $R_{K}+R_{1}$ which we can add together to show that
\begin{equation}
2\sum_{k=1}^{K}R_{k}\;\leq \sum_{\ell=2}^{K+1}\log(1+P(|h_{\ell,\ell-1}|^{2}+ |h_{\ell\ell}|^{2})) + K\overline R_{\rm b}
+o(\log(P)).\end{equation} Since $\overline R_{\rm b}\leq L(P)$ for any achievable scheme, dividing by $\log(P)$ and taking the limit as $P\rightarrow \infty$ yields $2\msf{DoF}(\alpha)\leq 1+\alpha$ as required. 
Further, by considering each user separately (single user bound) we can trivially obtain that $\msf{DoF}(\alpha)\leq 1$, and hence conclude that $$\msf{DoF}(\alpha) \leq \min\left\{1,\frac{1+\alpha}{2}\right\},\;\alpha\geq0$$ as stated by Theorem~\ref{thm:out}.\hfill \QED

\clearpage

\section{Proof of Theorem~\ref{thm:out2}: Transmitter Cooperation Upper Bound }\label{proof:conv2}

Here, following similar steps as in the receiver cooperation case, we will first bound all the rate pairs $R_{i}+R_{j}$ and then sum them up to obtain the corresponding result for $\sum R_{k}$.
Let $\tilde M_{i}^{[n]}\triangleq \left\{[m_{i\rightarrow  \ell}(\tau)]_{\tau=1}^{n}, \ell\in\{1,...,K\},\ell\neq i  \right\}$ denote all the backhhaul messages that originate from transmitter $i$. In order to bound $R_{1}+R_{2}$ we will consider a genie that gives 
$$W_{2},\;\;\tilde M^{[n]}_{1}\; \mbox{and}\;\;y_{2}^{n}\triangleq [y_{2}(\tau)]_{\tau=1}^{n}$$ to receiver one,  and the messages $$W_{[3:K]}\triangleq [W_{3},W_{4},...,W_{K}]$$ to both receivers one and two, as side information.  
Notice that the encoded signals $(x_{2}^{n}, x_{3}^{n},..., x_{K}^{n})$ are fully determined as a function of $(\tilde M^{[n]}_{1},W_{[2:K]})$, and hence, with this side information, receiver~1 will eventually be able  to eliminate all interference.

Starting from Fano's inequality we have that

\begin{align}
&n(R_{1}+R_{2}-\epsilon_{n}) \nonumber
\\
&\stackrel{\phantom{\rm (0)}}{=} I\big(W_{1};y_{1}^{n}\big) + 
I\big(W_{2};y_{2}^{n}\big)\nonumber
\\
&\stackrel{\phantom{\rm (0)}}{\leq} 
I\big(W_{1};y_{1}^{n},y_{2}^{n},{\tilde M}_{1}^{[n]}\big|W_{2},W_{[3:K]}\big) + 
I\big(W_{2},{\tilde M}_{1}^{[n]};y_{2}^{n}\big|W_{[3:K]}\big)\nonumber
\\
&\stackrel{\phantom{\rm (0)}}{=} 
I\big(W_{1};{\tilde M}_{1}^{[n]}\big|W_{2},W_{[3:K]}\big) + I\big(W_{1};y_{1}^{n},y_{2}^{n}\big|{\tilde M}_{1}^{[n]},W_{2},W_{[3:K]}\big) + 
I\big(W_{2},{\tilde M}_{1}^{[n]};y_{2}^{n}\big|W_{[3:K]}\big)\nonumber
\\ 
&\stackrel{\phantom{\rm (0)}}{\leq} 
H({\tilde M}_{1}^{[n]}) + I\big(W_{1};y_{1}^{n},y_{2}^{n}\big|{\tilde M}_{1}^{[n]},W_{2},W_{[3:K]}\big) + 
I\big(W_{2},{\tilde M}_{1}^{[n]};y_{2}^{n}\big|W_{[3:K]}\big)\nonumber
\\
&\stackrel{\phantom{\rm (0)}}{\leq}  
\sum_{j\neq1}H([m_{1\rightarrow j}(\tau)]_{\tau=1}^{n})+ I\big(W_{1};y_{1}^{n},y_{2}^{n}\big|{\tilde M}_{1}^{[n]},W_{2},W_{[3:K]}\big) + 
I\big(W_{2},{\tilde M}_{1}^{[n]};y_{2}^{n}\big|W_{[3:K]}\big)\nonumber
\\
&\stackrel{\phantom{\rm (0)}}{=}  
n\sum_{j\neq1}R_{\rm b}^{[1,j]} + I\big(W_{1};y_{1}^{n},y_{2}^{n}\big|{\tilde M}_{1}^{[n]},W_{2},W_{[3:K]}\big) + 
I\big(W_{2},{\tilde M}_{1}^{[n]};y_{2}^{n}\big|W_{[3:K]}\big)\label{lasteq}
\end{align}

The last two terms in (\ref{lasteq}) can be further bounded as
\begin{align}
&I\big(W_{1};y_{1}^{n},y_{2}^{n}\big|{\tilde M}_{1}^{[n]},W_{2},W_{[3:K]}\big) + 
I\big(W_{2},{\tilde M}_{1}^{[n]};y_{2}^{n}\big|W_{[3:K]}\big)\nonumber\\
&\stackrel{\phantom{\rm (0)}}{=} 
h\big(y_{1}^{n}\big|y_{2}^{n},{\tilde M}_{1}^{[n]},W_{[2:K]}\big) - h\big(y_{1}^{n},y_{2}^{n}\big|{\tilde M}_{1}^{[n]},W_{[1:K]}\big)  + h\big(y_{2}^{n}\big|W_{[3:K]}\big) \nonumber
\\
&\stackrel{\phantom{\rm (0)}}{\leq} 
h\big(h_{11}x_{1}^{n}+z_{1}^{n}\big|h_{21}x_{1}^{n}+z_{2}^{n}\big) - h\big(z_{1}^{n},z_{2}^{n}\big)  + h\big(y_{2}^{n}\big) \label{sseqs}
\\
&\stackrel{\phantom{\rm (0)}}{\leq} n\log\big(1+|h_{11}|^{2}/|h_{21}|^{2}\big) + n\log\big(1+P\textstyle\sum_{i,j} |h_{2i}||h_{2j}|\big)\label{sseqs2}
\end{align}
where (\ref{sseqs}) follows from the fact that $x_{[2:K]}^{n}$ is a function of $({\tilde M}_{1}^{[n]},W_{[2:K]})$ and (\ref{sseqs2}) from the fact that the Gaussian distribution maximizes the differential entropy for a given variance and by applying Lemma~\ref{lemma1} that is stated below.

\begin{lemma}\label{lemma1}
Let $\xv_{k}\in \CC^{n}$, $k\in \{1,2,...,K\}$, be any random vectors satisfying $\frac{1}{n}\EE[\xv_{k}^{\rm H}\xv_{k}]\leq P_{k}$, $\forall k$, and let $\sv \triangleq \sum_{k=1}^{K}\xv_{k}$ and $\Km_{\sv} \triangleq \EE\left[ (\sv - \EE[\sv])(\sv - \EE[\sv])^{\rm H}   \right]$. We have that 
$$\mbox{\rm det}({\bf I} + \Km_{\sv})^{{1}/{n}}\leq 1 + \sum_{k=1}^{K}\sum_{\ell=1}^{K}\sqrt{P_{k}P_{\ell}}.$$
 
\end{lemma}
\begin{proof}
Since ${\bf I} + \Km_{\sv}$ is positive definite, we have (from the Hadamard's inequality followed by the arithmetic-geometric mean inequality) that  
\begin{equation}
\mbox{\rm det}({\bf I} + \Km_{\sv})^{{1}/{n}}\leq \frac{1}{n}\mbox{tr}({\bf I} + \Km_{\sv}).\label{coveq}\end{equation}
Further, we can rewrite
\begin{align}
\frac{1}{n}\mbox{tr}({\bf I} + \Km_{\sv}) &= \frac{1}{n}\sum_{i=1}^{n}\left(1 + \mbox{\sf Var}[\sv(i)]\right) \\
&= 1+ \frac{1}{n}\sum_{i=1}^{n}\mbox{\sf Var}\left[\sum_{k=1}^{K}\xv_{k}(i)\right]\\
&= 1+ \frac{1}{n}\sum_{i=1}^{n}\sum_{k=1}^{K}\sum_{\ell=1}^{K}\mbox{\sf Cov}\left[\xv_{k}(i),\xv_{\ell}(i)\right]\\
&=1+ \sum_{k=1}^{K}\sum_{\ell=1}^{K}\frac{1}{n}\sum_{i=1}^{n}\mbox{\sf Cov}\left[\xv_{k}(i),\xv_{\ell}(i)\right]\label{temp1}
\end{align}
Now we can bound
\begin{align}
\sum_{i=1}^{n}\mbox{\sf Cov}\left[\xv_{k}(i),\xv_{\ell}(i)\right]
&\leq \sum_{i=1}^{n}\sqrt{\mbox{\sf Var}[\xv_{k}(i)] \cdot\mbox{\sf Var}[\xv_{\ell}(i)] }
\\
&\leq 
\sqrt{\left(\sum_{i=1}^{n}\mbox{\sf Var}[\xv_{k}(i)]\right)\cdot \left(\sum_{i=1}^{n}\mbox{\sf Var}[\xv_{\ell}(i)]\right)} \\
&\leq
\sqrt{\left(\sum_{i=1}^{n}\EE\left[|\xv_{k}(i)|^{2}\right]\right)\cdot \left(\sum_{i=1}^{n}\EE\left[|\xv_{\ell}(i)|^{2}\right]\right)} 
\\
&= 
\sqrt{\left(\EE[\xv_{k}^{\rm H}\xv_{k}]\right)\cdot \left(\EE[\xv_{\ell}^{\rm H}\xv_{\ell}]\right)}  \\
&\leq n\sqrt{P_{k}P_{\ell}}. \label{temp2}
\end{align}
And hence, substituting (\ref{temp2}) in (\ref{temp1}) and (\ref{coveq}) yields the desired result and completes the proof.
%
\end{proof}

Putting everything together we have that
\begin{equation}
R_{1}+R_{2} \leq \log\big(1+|h_{11}|^{2}/|h_{21}|^{2}\big) + \log\big(1+P\textstyle\sum_{i,j} |h_{2i}h_{2j}^{*}|\big) + \sum_{j\neq1}R_{\rm b}^{[1,j]}.
\end{equation}

In a similar way, we can obtain bounds of the same form for the pairs $R_{2}+R_{3}$, $R_{3}+R_{4}$, up to $R_{K-1}+R_{K}$ and $R_{K}+R_{1}$ which we can add together to show that
\begin{equation}
2\sum_{k=1}^{K}R_{k}\;\leq \sum_{k=1}^{K}\log\big(1+P\textstyle\sum_{i,j} |h_{k i}h_{k j}^{*}|\big) + K\overline R_{\rm b}
+o(\log(P)).\end{equation} Since $\overline R_{\rm b}\leq L(P)$ for any achievable scheme, dividing by $\log(P)$ and taking the limit as $P\rightarrow \infty$ yields $2\msf{DoF}(\alpha)\leq 1+\alpha$ as required. 
Further, by considering each user separately (single user bound) we can trivially obtain that $\msf{DoF}(\alpha)\leq 1$, and hence conclude that $$\msf{DoF}(\alpha) \leq \min\left\{1,\frac{1+\alpha}{2}\right\},\;\alpha\geq0$$ as stated by Theorem~\ref{thm:out2}.\hfill \QED

\clearpage

\clearpage

\bibliographystyle{ieeetr}
\bibliography{referencesIT}

\end{document}